\documentclass[journal]{IEEEtran}
\usepackage{amsmath,amsfonts}
\usepackage{algorithmicx}
\usepackage[ruled,vlined,linesnumbered]{algorithm2e}
\usepackage{array}
\usepackage{textcomp}
\usepackage{stfloats}
\usepackage{url}
\usepackage{verbatim}
\usepackage{graphicx}
\usepackage{cite}
\usepackage[dvipsnames]{xcolor}
\usepackage{multirow}
\usepackage{float}
\usepackage{adjustbox}
\usepackage{soul}
\usepackage{mathtools}
\usepackage{flushend}
\usepackage{enumitem}
\usepackage{subfigure}

\usepackage{amsthm}
\usepackage{colortbl}

\newtheorem{Definition}{Definition}
\newtheorem{Theorem}{Theorem}
\newtheorem{Remark}{Remark}

\usepackage[many]{tcolorbox}  
\newtcolorbox{boxD}{
    arc = 0mm,
    boxrule = 0pt, 
    toprule = 0pt, 
    bottomrule = 0pt, 
    colframe = black, 
    colback = white,
    float, floatplacement=t
}

\SetKwRepeat{Do}{do}{while}

\def\change#1{{\textcolor{blue}{#1}}}

\begin{document}

\title{Safe-by-Construction Autonomous Vehicle Overtaking using Control Barrier Functions and Model Predictive Control}

\author{Dingran Yuan, Xinyi Yu, Shaoyuan Li, Xiang Yin
\thanks{This work was supported by the National Key Research and Development Program of China (2018AAA0101700) and the National Natural Science Foundation of China (62061136004, 61803259, 61833012). }
\thanks{The authors are with the Department of Automation, Shanghai Jiao Tong University, Shanghai, 200240 China (e-mail: dingran.yuan@foxmail.com; yuxinyi-12@sjtu.edu.cn; syli@sjtu.edu.cn; yinxiang@sjtu.edu.cn). }}

\markboth{}{}

\IEEEpubid{}

\maketitle

\begin{abstract}
Ensuring safety for vehicle overtaking systems is one of the most fundamental and challenging tasks in autonomous driving. This task is particularly intricate  when the vehicle must not only overtake its front vehicle safely but also  consider the  presence of potential opposing vehicles in the opposite lane that it will temporarily occupy. 
In order to tackle the overtaking task in such challenging scenarios, we introduce a novel integrated framework tailored for vehicle overtaking maneuvers. 
Our approach integrates the theories of varying-level control barrier functions (CBF) and time-optimal model predictive control (MPC).  
The main feature of our proposed overtaking strategy is that it is safe-by-construction, which enables rigorous mathematical proof and validation of the safety guarantees. 
We show that the proposed framework is applicable when the opposing vehicle is either fully autonomous or driven by human drivers. To demonstrate our framework, we perform a set of simulations for overtaking scenarios under different settings. The simulation results show the superiority of our framework in the sense that it ensures collision-free and achieves better safety performance compared with  the standard MPC-based approach without safety guarantees.
\end{abstract}


\section{Introduction}

\subsection{Motivation}

\IEEEPARstart{O}{ver} the past decade, smart vehicle control systems have garnered significant attention due to advancements in real-time  computing and advanced control theory. The benefits of implementing smart controls have been widely acknowledged, including improved traffic arrangements, reduced accidents, and enhanced driving comfort  \cite{autiliCooperativeIntelligentTransport2021,xuDecentralizedTimeEnergyoptimal2022,wangAdaptiveDynamicPath2023,zhouSpatiotemporalFeatureEncoding2022}. Specifically, high-level smart vehicle control can generally be classified into four main categories: adaptive cruise control (ACC), lane keeping (LK), lane change (LC), and \emph{overtaking}. Of these categories, overtaking maneuvers pose the particular challenge as they typically need to interact with other vehicles.
 
One of the most challenging scenarios in vehicle overtaking maneuvers is probably  the two-way overtaking scenario, as depicted in Figure~\ref{fig.overtaking_problem}. In this scenario, the vehicle must not only overtake its front vehicle safely but also consider the potential presence of the opposing vehicles in the opposite lane that it will temporarily occupy. According to a study conducted by the German Insurers Accident Research (UDV), the accident rate (measured in accidents per $10^6$ vehicles km) of all types of vehicles during overtaking on a two-way lane is approximately 50\% higher than that on a one-way lane without opposing traffic \cite{richterCausesConsequencesCountermeasures2017}. This further underscores the importance and difficulty of addressing the two-way road overtaking problem.

As  human-centered   systems, autonomous vehicles are \emph{safety-critical}, which necessitates a high level of safety and robustness. As such, there is an increasing demand for \emph{formal safety guarantees} in the analysis and design of autonomous driving strategies. Particularly, instead of relying solely on experimental testing to assert safety, there is a need to provide \emph{provably safe} evidence in a correct-by-construction manner. This is particularly relevant in the context of autonomous overtaking systems, where existing studies have utilized either model-based methods \cite{usmanAutonomousVehicleOvertaking2009,petrovModelingNonlinearAdaptive2014,chaeVirtualTargetBasedOvertaking2020} or data-driven methods \cite{milanesIntelligentAutomaticOvertaking2012, liDeepReinforcementLearning2020, guptaNovelGraphbasedMachine2022, liCombiningDecisionMaking2022} to design vehicle overtaking maneuvers. 
However, most existing methods can hardly provide rigorous safety guarantees when operating in dynamic environments with changing surroundings. Therefore, there is a pressing need for research in this field to develop methods that can guarantee safety in such scenarios.

\IEEEpubidadjcol

\subsection{Our Contributions}

\begin{figure} [!ht]
    \centering
    \includegraphics[width = 0.9\linewidth]{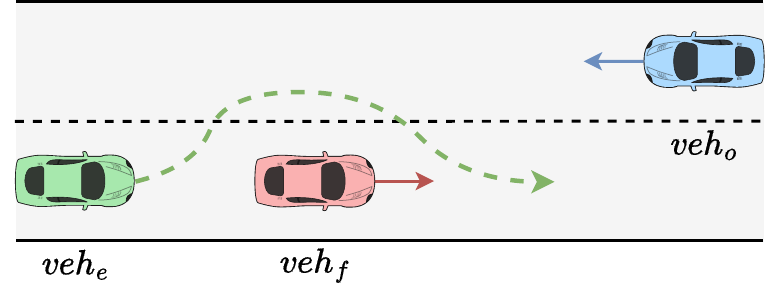}
    \caption{A two-way road overtaking  scenario with possible opposing traffic.}
    \label{fig.overtaking_problem}
\end{figure}

Our work aims to address the gap between the present non-verifiable and error-prone control designs for overtaking control systems in autonomous driving and the increasing needs for safety guarantees. Specifically, we consider a challenging two-way overtaking scenario that includes possible opposite vehicles, as illustrated in Figure~\ref{fig.overtaking_problem}. To tackle this problem, we develop a novel framework that integrates the theories of varying-level control barrier functions (CBF) and time-optimal model predictive control (MPC). The main feature of our proposed overtaking strategy is that it is \emph{safe-by-construction}, meaning that we can mathematically prove its safety guarantees. This approach eliminates the need for the possible tedious subsequential test-and-redesign cycle that often arises when formal methods are not used.

More specifically, the main contributions of this paper are summarized as follows:
\begin{itemize}[leftmargin=*] 
    \item 
    First, we propose a novel notion, referred to as varying-level control barrier functions (VL-CBF), which offers adjustable level sets, in contrast to standard CBFs that rely on fixed level sets. The proposed VL-CBFs are particularly suitable for autonomous driving applications, as they enable the incorporation of the driver's preferences and facilitate the online adaptation of relevant parameters.
    \item 
    Then we provide an integrated framework that combines the proposed VL-CBF  with model predictive control for vehicle overtaking in a more complex scenario that includes possible opposing traffic. Specifically, at each instant, our approach solves two optimization problems: one for overtaking the front vehicle and the other for returning to the original position when the former problem cannot be solved.   
    \item 
    We show that our proposed framework guarantees safety, and we formally establish this by mathematically proving that at each instant, either the forward or the backward problem has a solution. Also, we show that any planned trajectories and associated control laws will prevent the ego vehicle from colliding with either the front or opposing vehicle.  A set of simulations are also conducted to demonstrate our results. 
\end{itemize}

\subsection{Related Works}

\subsubsection{Autonomous Overtaking Maneuvers}
The increasing computational power has led to a rapid development of on-road overtaking controllers in the past years \cite{liCombiningDecisionMaking2022, xieDistributedMotionPlanning2022,dixitTrajectoryPlanningAutonomous2019}. Approaches in the literature for tackling the vehicle overtaking problem can be broadly categorized into three types: traditional model-based methods, data-driven methods, and learning-based methods. 
For example, in \cite{petrovModelingNonlinearAdaptive2014}, Petrov \textit{et al.} introduced a three-phase overtaking model which decomposes the overtaking actions into  consecutive LC, LK and LC sub-problems, then using a kinematic model-based adaptive controller to solve the overall problem. 
This approach works under a simple scenario without surrounding traffic. In \cite{diazalonsoLaneChangeDecisionAid2008,chaeVirtualTargetBasedOvertaking2020}, the authors used extended Kalman filters to estimate the vehicle states in order to switch between LC and LK modes. 

Data-driven approaches were also developed by researchers to synthesize control laws from data collected from expert drivers and simulations  \cite{liNovelUAVenabledData2020, naranjoLaneChangeFuzzyControl2008,milanesIntelligentAutomaticOvertaking2012,wangCollisionAvoidanceSystem2017,doHumanDriversBased2017}. In  \cite{wangCollisionAvoidanceSystem2017},  vehicle trajectories are generated based on the dangerous level computed by a fuzzy inference system developed with naturalistic driving data. 
In \cite{doHumanDriversBased2017}, a two-segment lane change controller  
was proposed  in order to   mimic  the behavior of human-drivers. However, the accuracy of these methods heavily depends on the  data sets, and they cannot provide rigorous guarantees for safety. 
Learning-based planning methods with safety considerations have also been proposed recently in \cite{liuPhysicsawareSafetyassuredDesign2022,liuSafetydrivenInteractivePlanning2023}. For example, in \cite{liuSafetydrivenInteractivePlanning2023}, the authors proposed a neural network-based LC controller that leverages a back-up strategy to enhance safety.
The readers are referred to \cite{chenInterpretableEndtoendUrban2021} for more planning methods within the end-to-end training framework.

Reinforcement learning techniques have also been used to  perform   overtaking  tasks in partially-known or unknown environments \cite{ngaiMultipleGoalReinforcementLearning2011,liangCIRLControllableImitative2018,liDeepReinforcementLearning2020}. 
For example,  Rosolia \textit{et al.} \cite{rosoliaLearningHowAutonomously2020} proposed a deep learning-based race  strategy  by combining with MPC controllers. 
However,  since learning-based methods can hardly cover all possible scenarios that may happen on road due to the nature of their finite training data, they still cannot provide formal safety guarantee\change{s} for overtaking tasks. 


\subsubsection{CBF-Based  Vehicle Control with Formal Guarantees}
In recent years, control barrier functions (CBFs) have emerged as a highly successful technique for ensuring the safety of autonomous systems with formal guarantees \cite{amesControlBarrierFunctions2019}. The fundamental principle of CBFs is to represent the safety of the system using a super-level set, whose forward invariance can be enforced through constraints on the control inputs \cite{amesControlBarrierFunctions2019}. In the pioneering work by Ames \textit{et al.} \cite{amesControlBarrierFunction2017}, quadratic programming (QP) was applied to CBF, which made CBF computationally feasible for real-time control.
To address various challenges in system design, many variants of CBFs have been proposed recently. For example, exponential CBF and high-order CBF have been developed to overcome the non-existence of first-order Lie derivatives \cite{nguyenExponentialControlBarrier2016, xiaoHighOrderControlBarrier2022, tanHighOrderBarrierFunctions2022}. Feasibility and discretization issues have been addressed in the literature, with works such as \cite{agrawalDiscreteControlBarrier2017, xiaoAdaptiveControlBarrier2022, xiongDiscreteTimeControlBarrier2022, choiRobustControlBarrier2021}. Additionally, observation-based CBFs \cite{dawsonLearningSafeGeneralizable2022} and decentralized CBFs \cite{qinLearningSafeMultiAgent2021} have been proposed to extend the generality and scalability for large-scale systems in unknown environments.
Barrier functions have also been utilized to enhance safety  for learning-based  controllers; see, e.g.,\cite{wangJointDifferentiableOptimization2023,xiaoDifferentiableControlBarrier2022,yangDifferentiableSafeController2023}. 
In \cite{wangJointDifferentiableOptimization2023}, the authors proposed a joint differentiable optimization and verification framework that can synthesize controllers and their certificates simultaneously. 
In \cite{xiaoDifferentiableControlBarrier2022,yangDifferentiableSafeController2023}, CBFs are used  as an add-on layer to  neural network-based controllers to ensure safety.

In the specific context of autonomous vehicles, CBFs have been successfully applied to ensure safety for four categories of basic tasks. For instance, Xu \textit{et al.}  designed a composed CBF through a contract-based method that provided safety guarantees for simultaneous LK and ACC \cite{xuCorrectnessGuaranteesComposition2018}. Lyu \textit{et al.}  focused on safe ramp merging cases, proposing a probabilistic CBF control law under Gaussian uncertainty of motion \cite{lyuAdaptiveSafeMerging2022}. Furthermore, in  \cite{lyuProbabilisticSafetyAssuredAdaptive2021}, the authors developed an adaptive CBF learning structure for the ego vehicle based on the observation of others. 
Particularly, for the vehicle overtaking problems,  He \textit{et al.}  proposed a rule-based three-phase overtaking framework using a finite state machine (FSM) to switch between different CBFs during different overtaking phases  \cite{heRuleBasedSafetyCriticalControl2021}. The  authors also proposed a race car overtaking strategy that combines CBF with a learning-based MPC method \cite{heAutonomousRacingMultiple2022}, enabling safe competition with surrounding vehicles on a closed track.
While these works provide basic solutions for formal safety guarantees in overtaking problems, it should be noted that the frameworks are developed only for simple scenarios without possible opposing traffic. For practical purposes, it is necessary to take unexpected opposing vehicles into account to ensure  safety.

\subsubsection{Challenges in  Overtaking with Opposing Vehicles}
\label{sec.Related-works-summary}
As we discussed above, most existing works on autonomous overtaking do not provide formal safety guarantees for the entire process except \cite{heRuleBasedSafetyCriticalControl2021,heAutonomousRacingMultiple2022}. 
However, \cite{heRuleBasedSafetyCriticalControl2021,heAutonomousRacingMultiple2022} consider the overtaking problem \emph{without} opposing vehicles. This scenario is fundamentally easier to handle than the case with opposing vehicles as we consider in this work. Specifically, without opposing vehicles, the completion of the overtaking task can be postponed to an infinite horizon as the ego vehicle can always stay in the adjacent lane. However, in the presence of opposing vehicles, the optimization problem has to be solved within a finite horizon since the opposing vehicle is approaching towards the ego vehicle.  Therefore, we need to provide the estimated time for accomplishing the overtaking task with safety guarantees. Furthermore, the presence of opposing vehicle may make the overtaking task   infeasible, 
and therefore, we also need to provide a back-up plan for safely returning to the original point. All these issues have not been considered in the literature, and our work aims to provide an integrated method to address these challenges.

\subsection{Paper Organizations}
The paper is organized as follows. 
Section~\ref{sec.problem} formulates the two-way road overtaking problem. 
Section~\ref{sec.preliminaries} revisits the theory of control barrier functions, and introduces the model we use in this study. Section~\ref{sec.VL-CBF}  proposes the varying-level control barrier functions and investigates its related properties. Section~\ref{sec.framework} introduces our   framework for two-way road overtaking, and establishes the formal guarantee for   safety. Section~\ref{sec.sims} provides a set of simulations to illustrate the proposed approach, and gives comparison of the baseline method. Finally, Section~\ref{sec.conclusion}  summarizes the paper  and discusses some future directions.

\section{Problem Formulation}
\label{sec.problem}
In this paper, we focus on an overtaking scenario on a two-way road as depicted in Figure~\ref{fig.overtaking_problem}. 
Specifically, the current lane with the same traffic direction is denoted by $\mathcal{L}_{\text{ego}}$  and the opposing lane with the opposite traffic direction is denoted by  $\mathcal{L}_{\text{opp}}$.
We term our controlled vehicle as the \textit{ego vehicle} (shown in green), denoted as $veh_e$; the front overtaken vehicle termed as \textit{overtaken vehicle} (shown in red), denoted as $veh_f$, where $f$ stands for ``front"; and the possible opposing vehicle (shown in blue) is denoted as $veh_o$. 

Our objective is to design an \emph{overtaking control strategy} for the ego vehicle $veh_e$ with \emph{provable safety guarantees}. Specifically, we require that:
\begin{itemize}[leftmargin=*] 
    \item 
    The ego vehicle $veh_e$ can safely overtake the front vehicle $veh_f$ by  temporarily occupying the opposite lane $\mathcal{L}_{\text{opp}}$, while ensuring collision avoidance with both $veh_f$ and the opposing vehicle $veh_o$  if the overtaking task is feasible. 
     \item
     If the overtaking task is not feasible, then ego vehicle  $veh_e$  must be able to return to its original position in  $\mathcal{L}_{\text{ego}}$ safely ensuring collision avoidance with other vehicles.
     \item 
     A formal guarantee can be established for collision avoidance and solution feasibilities throughout the entire overtaking process.
\end{itemize}


\begin{Remark} 
For the sake of simplicity and without loss of generality, in our later developments, we make the following assumptions regarding the behaviors of the front vehicle $veh_f$ and the opposing vehicle $veh_o$. 
\begin{itemize}[leftmargin=*] 
    \item 
     For the front vehicle  $veh_f$, 
     we assume that it is non-accelerating during the overtaking process. 
     This assumption is reasonable for a rational driving scenario.  
     In fact, our approach can be extended to the case,  where $veh_f$  is accelerating or non-accelerating but not interchangeably.
      \item
    For the opposing vehicle $veh_o$,
    our main analysis will be based on the assumption that it is also autonomous following the same safety-critical control law.  This assumption enables us to perform a thorough analysis. When the opposing vehicle is driven by a human driver, we will also discuss using conservative analysis to address any deviations from our assumptions in Section~\ref{sec.framework}. 
\end{itemize}  
\end{Remark}

\section{Preliminaries}
\label{sec.preliminaries}

\subsection{Vehicle Models}

In this work, the general dynamics of vehicles are described by the following nonlinear control affine model:
\begin{equation}
    \dot{\mathbf{x}} = f(\mathbf{x}) + g(\mathbf{x})\mathbf{u},
    \label{eq.control_affine_model}
\end{equation}
where $\mathbf{x}\in\mathcal{X}\subset\mathbb{R}^n$, and $\mathbf{u}\in\mathcal{U}\subset\mathbb{R}^m$ are the system state  and control input, respectively,  with $\mathcal{X}$ and $\mathcal{U}$ be the state and control constraints. Mappings $f$ and $g$ are both assumed to be Lipschitz continuous. For any initial condition $\mathbf{x}(0)\in\mathbb{R}^n$, there exists a maximum time interval $\mathcal{I}_{\text{max}}\coloneqq [0,\tau_{\text{max}})$ such that $\mathbf{x}(t)$ is the unique solution to \eqref{eq.control_affine_model} on $\mathcal{I}_{\text{max}}$.


More specifically, for the ego vehicle,  we consider a precise  non-holonomic kinematic bicycle model   \cite{kongKinematicDynamicVehicle2015}, which can be transferred into the following control affine model   \cite{heRuleBasedSafetyCriticalControl2021}:
\begin{equation}
    \dot{\mathbf{x}} = 
    \left[\begin{array}{c}
        \dot{x}\\
        \dot{y}\\
        \dot{\psi}\\
        \dot{v}  
    \end{array}\right] = \left[\begin{array}{c}
        v\cos\psi\\
        v\sin\psi\\
        0\\
        0
    \end{array}\right] + \left[\begin{array}{cc}
        0 & -v\sin\psi\\
        0 & v\cos\psi\\
        0 & \frac{v}{l_r}\\
        1 & 0
    \end{array}\right] \left[\begin{array}{c}
        \alpha\\
        \beta
    \end{array}\right], 
    \label{eq.kinematic_bicycle}
\end{equation} 
where 
$\mathbf{x} = [x,y,\psi, v]^\top$ denotes the state of the vehicle
and  $\mathbf{u} = [\alpha,\beta]^\top$ denotes the input of the vehicle with $\alpha$ be the acceleration at vehicle's center of gravity (c.g.) and $\beta$ be the slip angle of the vehicle. The system states $x,y$ are the longitudinal and lateral positions of the c.g. with respect to the inertial frame, respectively, 
and $\psi, v$ represent  the orientation and forward velocity of the vehicle, respectively. In the rest of this paper, we will use normal letter and bold letter to represent scalar and vector, respectively.

For the overtaken and opposing vehicles,  lateral movements are negligible during the overtaking process, and only longitudinal behaviors are of our interest. Therefore, we use a  simpler double integral model described as follows:
\begin{equation}
    \dot{\mathbf{x}} = 
    \left[\begin{array}{c}
        \dot{x}\\
        \dot{v}
    \end{array}\right] = \left[\begin{array}{cc}
        0 & 1\\
        0 & 0
    \end{array}\right] \left[\begin{array}{c}
        x\\
        v
    \end{array}\right] + \left[\begin{array}{c}
        0\\
        1
    \end{array}\right]\alpha,
    \label{eq.double_integrator}
\end{equation}
where $\alpha$  represents the along-road longitudinal acceleration rate of the double integral model.


\subsection{Control Barrier Functions}
Control barrier functions (CBFs) are recently developed successful techniques for ensuring safety for  constrained control systems \cite{amesControlBarrierFunction2017, amesControlBarrierFunctions2019, xiaoSufficientConditionsFeasibility2022}. 
Following the standard  definitions  \cite{amesControlBarrierFunction2017}, we assume that the \emph{safe set} $\mathcal{C}$ of the system is represented by the super-level set of a time-varying differentiable  function $h(\mathbf{x},t)$, i.e.,
\begin{equation} 
        \mathcal{C}(t) \coloneqq \{\mathbf{x}\in\mathbb{R}^n: h(\mathbf{x},t)\geq 0\}.
        \label{eq.safe_set}
\end{equation}
Note that when the control law $\mathbf{u} \coloneqq k(\mathbf{x})$ is locally Lipschitz continuous, the   control affine model $\dot{\mathbf{x}} = f(\mathbf{x}) + g(\mathbf{x})k(\mathbf{x})$ is still locally Lipschitz continuous. Therefore, for any initial condition $\mathbf{x}(0)\!\in\!\mathcal{C}(0)$, system  \eqref{eq.control_affine_model} always has a unique solution within a time interval. 
The safety of the system is defined by the following forward invariance property.  
\begin{Definition}[Forward Invariance and Safety\cite{xiaoHighOrderControlBarrier2022}]\upshape
    The set $\mathcal{C}$ is said to be \emph{forward invariant} for a given control law $\mathbf{u}$ if for every $\mathbf{x}(0)\in\mathcal{C}(0)$, we have $\mathbf{x}(t)\in\mathcal{C}(t)$ holds for all $t\in[0,\tau_{\text{max}})$,  where $\mathbf{x}(t)$ is the unique solution to   \eqref{eq.control_affine_model} starting from  $\mathbf{x}(0)$. 
    We say the system is \emph{safe} if safe set $\mathcal{C}$ is forward invariant. 
\end{Definition}



We use the concept of control barrier functions to characterize the admissible set of control inputs that guarantee the set $\mathcal{C}$ is forward invariant. Recall that a continuous function $\kappa:(-b,a)\to(-\infty,\infty)$ for $a, b \!\in\! \mathbb{R}_{\geq 0}$ belongs to the extended class of $\mathcal{K}$ functions (also called class $\mathcal{K}_e$ functions) if it is monotonically increasing and satisfies $\kappa(0) = 0$ \cite{khalilNonlinearSystems2015}. Now we introduce the definition of CBF as follows.

\begin{Definition}[Control Barrier Functions \cite{amesControlBarrierFunction2017}]\label{def:CBF}\upshape
Let $h(\mathbf{x},t):\mathbb{R}^n\times \mathcal{I}_{\text{max}}\to\mathbb{R}$ be a continuously differentiable function. 
We say  $h$ is a \emph{control barrier function} for  system  \eqref{eq.control_affine_model} if there exists a class $\mathcal{K}_e$ function $\kappa(\cdot)$ such that 
\begin{equation}
    \sup_{\mathbf{u}\in\mathcal{U}}\left[L_fh(\mathbf{x},t) + L_gh(\mathbf{x},t)\mathbf{u} + \frac{\partial h}{\partial t}\right]\geq -\kappa(h(\mathbf{x},t)),
    \label{eq.def_cbf}
\end{equation}
for all $(\mathbf{x},t)\!\in\!\mathcal{C}(t)\times \mathcal{I}_{\text{max}}$, where $L_fh(\mathbf{x},t) \coloneqq \frac{\partial h(\mathbf{x},t)^\top}{\partial \mathbf{x}}f(\mathbf{x})$ and $L_gh(\mathbf{x},t) \coloneqq \frac{\partial h(\mathbf{x},t)^\top}{\partial \mathbf{x}}g(\mathbf{x})$ are the Lie derivatives. 
\end{Definition}

Given a CBF $h(\mathbf{x},t)$, the set of all control values that render set $\mathcal{C}$ safe is given by
\begin{equation}
        K_{\text{cbf}} (\mathbf{x},t)\! \coloneqq\! \left\{\mathbf{u}\in\mathcal{U}:L_f h\! +\! L_g h \mathbf{u}\! +\! \frac{\partial h}{\partial t} \! +\!\kappa(h)\! \geq \! 0\right\},
    \label{eq.K_cbf}
\end{equation}
which denotes the safe control input set. The following theorem shows that the existence of a CBF can provide formal safety guarantee for the system. 

\begin{Theorem}[\!\! \cite{amesControlBarrierFunction2017}]\upshape
Let $h(\mathbf{x},t)$ be a valid CBF and assume that $\frac{\partial h(\mathbf{x},t)}{\partial \mathbf{x}}\neq \mathbf{0}^n$ for all $(\mathbf{x},t)\in \partial \mathcal{C}(t)\times\mathcal{I}_{\text{max}}\subset\mathbb{R}^n\times \mathbb{R}$. Then any Lipschitz continuous controller $\mathbf{u} \coloneqq k(\mathbf{x})$ such that $\mathbf{u} \in K_{\text{cbf}}(\mathbf{x},t)$ for every $(\mathbf{x},t)\in\mathcal{C}(t)\times \mathcal{I}_{\text{max}}$ will render set $\mathcal{C}$ forward invariant. 
\end{Theorem}


 For control affine systems in \eqref{eq.control_affine_model}, $K_{\text{cbf}}$ in \eqref{eq.K_cbf} leads to linear constraints on  control inputs, which can be incorporated into computationally tractable QP-based controllers \cite{amesControlBarrierFunction2017}. When incorporating different driving preferences in on-road overtaking scenarios, conventional methods require us to redesign $h$ case-by-case, whereas we introduce our VL-CBF method to provide adjustable level set of $h$ in the next section.  


\section{Varying-Level Control Barrier Functions}
\label{sec.VL-CBF}

In the preceding section, we presented a fundamental definition of control barrier functions  as a crucial tool for ensuring safety. In this section, we introduce a novel concept called \emph{Varying-Level Control Barrier Functions}. This concept is aimed at improving the explanatory power and user parametrization capabilities of CBFs.

\subsection{Varying-Level Control Barrier Functions}
\label{sec.VL-CBF_properties}

Note that Definition~\ref{def:CBF} for CBF only requires the existence of a $\mathcal{K}_e$ function $\kappa(\cdot)$ for $h$. However, different functions essentially correspond to different autonomous system behavior patterns. In autonomous driving, control laws must characterize user-defined diverse behavior patterns and be explainable. To this end, the concept of parametric CBF (PCBF) was proposed in \cite{lyuAdaptiveSafeMerging2022} to capture these requirements, where the function $\kappa(\cdot)$ is restricted to a polynomial with odd powers of $h$. The coefficients in the polynomial provide parameters for the driver to choose, achieving the purpose of user-defined diverse behavior patterns.

However, PCBF can only enforce the \emph{zero super-level set} as described in \eqref{eq.safe_set}. As value function $h(\mathbf{x},t)$ is typically selected as the distance between two vehicles, enforcement of the zero super-level set essentially means that two vehicles can be arbitrarily close, as long as they do not collide. In practice, drivers are more conservative and want to enforce a \emph{positive super-level set} of $h(\mathbf{x},t)$. This conservative degree should also be parameterized, so that each user can choose their safety margin. To this end, motivated by the existing PCBF, we propose the following VL-CBF. 

\smallskip
\begin{Definition}[Varying-Level Control Barrier Functions]\upshape \label{def.VL-CBF}
Given a dynamic system model as in \eqref{eq.control_affine_model}, and the nominal safe set $\mathcal{C}(t)$ defined as \eqref{eq.safe_set}  with time-varying differentiable value function $h(\mathbf{x},t):\mathbb{R}^n\times\mathcal{I}_{\text{max}} \to\mathbb{R}$. 
Then $h$ is said to be a varying-level control barrier function (VL-CBF) of order $m\geq 1$ if
    \begin{equation}
        \sup_{ \mathbf{u}\in \mathcal{U}}[L_f h(\mathbf{x},t) + L_g h(\mathbf{x},t)\mathbf{u} + \frac{\partial h(\mathbf{x},t)}{\partial t}]
        \geq -\Lambda H(\mathbf{x},t),
        \label{eq.VL-CBF}
    \end{equation}
holds for all $(\mathbf{x},t)\!\in\!\mathcal{C}(t)\times \mathcal{I}_{\text{max}}$, 
where 
$\Lambda \!\coloneqq\! [\lambda_0,\lambda_1,\dots,\lambda_m]\!\in\!\mathbb{R}^{m+1}$
and 
$H(\mathbf{x},t) \coloneqq [1, h, h^3,h^5,\dots,h^{2m-1}]^\top$ 
such that 
$\Lambda H(\mathbf{x},t)$ is a polynomial equation consists of a constant and odd powers of  $h$.
\end{Definition}





\noindent The purpose of defining VL-CBF is to enforce an arbitrary driver-defined non-negative super-level set 
\begin{equation}
    \mathcal{C}_{\epsilon}(t) \coloneqq \{x\in\mathcal{X}: h(\mathbf{x},t) \geq \epsilon\},
\end{equation}
such that more conservative safety requirement can be fulfilled. 
Similar to the case of standard CBF, we define
\begin{equation}\label{eq:vlcbf-input}
    K_{\text{vlcbf}}(\mathbf{x},t) \!\coloneqq\! \left\{\mathbf{u}\in\mathcal{U}\!:\!L_f h \!+\! L_g h\mathbf{u} \!+\! \frac{\partial h}{\partial t} \!+\! \Lambda H \geq{0}\right\}
\end{equation}
as the set of safe control inputs. 
Then the following result shows that this purpose can be achieved by suitably choosing a negative parameter $\lambda_0$, which will be referred to as the \emph{level parameter}, in $\Lambda$. 

\begin{Theorem}\upshape
    \label{thm.VL-CBF_levelset}
Let $h(\mathbf{x},t)$ be a valid VL-CBF, and assume that $\frac{\partial h}{\partial\mathbf{x}}\neq \mathbf{0}^n$ for all $(\mathbf{x},t)\in\partial\mathcal{C}(t)\times \mathcal{I}_{\text{max}}\subset\mathbb{R}^n\times\mathbb{R}$. Then for any $\epsilon \in \mathbb{R}_{\geq{0}}$ 
there exists a set of parameter $\Lambda$, where the level parameter $\lambda_0\in\mathbb{R}_{\leq 0}$, and $\lambda_i\in\mathbb{R}_{\geq{0}}$ for $i\in\{1,\dots,m\}$, such that any Lipschitz continuous controller satisfying $\mathbf{u}\in K_{\text{vlcbf}}(\mathbf{x},t)$ for every $(\mathbf{x},t)\in\mathcal{C}_\epsilon(t)\times \mathcal{I}_{\text{max}}$ will render set $\mathcal{C}_\epsilon$ forward invariant.
\end{Theorem}
\begin{proof} 
        We define $\tilde{\kappa}(h) \coloneqq [\lambda_1,\dots,\lambda_m][h,\dots,h^{2m-1}]^{\top}$, which is a polynomial sum starting from the second term of the polynomial equation $\Lambda H$. As proved in \cite{lyuAdaptiveSafeMerging2022}, $\tilde{\kappa}(h)$ is a class $\mathcal{K}_e$ function of $h$. Using this new notation, \eqref{eq.VL-CBF} can be rewritten as
        \begin{equation}
            L_f h + L_g h\mathbf{u} + \frac{\partial h}{\partial t} \geq -\lambda_0 - \tilde{\kappa}(h),
            \label{eq.VL-CBF_proof}
        \end{equation}
        where the right hand side of the equation consists of a class $\mathcal{K}_e$ function and a level parameter $\lambda_0$. For any $\epsilon \in\mathbb{R}_{\geq{0}}$, the initial condition $\mathbf{x}(0)\in\mathcal{C}_\epsilon(0)$ implies $h(\mathbf{x}(0),0) \geq \epsilon$, then we show that there always exists a $\lambda_0$ that can render $\mathcal{C}_\epsilon$ forward invariant.
        
        Let $\lambda_0 = -\tilde{\kappa}(\epsilon)\in\mathbb{R}_{\leq{0}}$, through \eqref{eq.VL-CBF_proof} we have $\dot{h}\geq \tilde{\kappa}(\epsilon) - \tilde{\kappa}(h)$ for all $t\in\mathcal{I}_{\text{max}}$ whenever $h > \epsilon$, which suggests the value of $h$ might decrease when $h$ is greater than $\epsilon$; whereas, when $h\to\epsilon$ we have $\dot{h}\geq 0$, suggesting the value of $h$ can not decrease further below $\epsilon$. Thus $\forall \mathbf{x}(0)\in\mathcal{C}_\epsilon(0)$, we always have $ h(\mathbf{x},t)\geq \epsilon,\forall t\in\mathcal{I}_{\text{max}}$. This completes the proof. 
\end{proof}

As illustrated in Figure~\ref{fig.VL-CBF_graph},
we show a VL-CBF  parameterized by three parameters $\Lambda = [\lambda_0, \lambda_1, \lambda_2]$, under the same value function $h$. 
Specifically,  different choices of non-level parameters $\lambda_1$ and $\lambda_2$ yields different divergent horizons of its level set.
Furthermore, when those non-level parameters are fixed, 
choosing a smaller level parameter  $\lambda_0$ yields a higher the horizon of the super-level set for $h(\mathbf{x},t)$, which implies that a more conservative safety control strategy  enforcing a positive super-level set. The top blue curve illustrates the case when the initial value for $h$ is lower than the safety level, thus the value for $h$ will rise from the beginning.

\smallskip
\begin{Remark}
Note that the requirement of enforcing the $\epsilon$-super-level set of $h$ can also be achieved via the standard CBF by redefining function $h$. This, however, is a hand-coded approach, and one needs to change the underlying $h$ for different $\epsilon$. The main feature of the proposed VL-CBF is that we can avoid this hand-coding. The level parameter $\lambda_0$ gives the system the ability to change the horizon of its super-level set for $h(\mathbf{x},t)$, without the need to redesign the original $h(\mathbf{x},t)$ for different conservativeness degrees. This setting makes more sense for decentralized autonomous vehicle control, since under the identical safety criteria, each individual may have its own interpretation of the conservativeness towards the safety boundary.
\end{Remark}
\smallskip

\begin{figure}
    \centering
    \includegraphics[width = 0.89\linewidth]{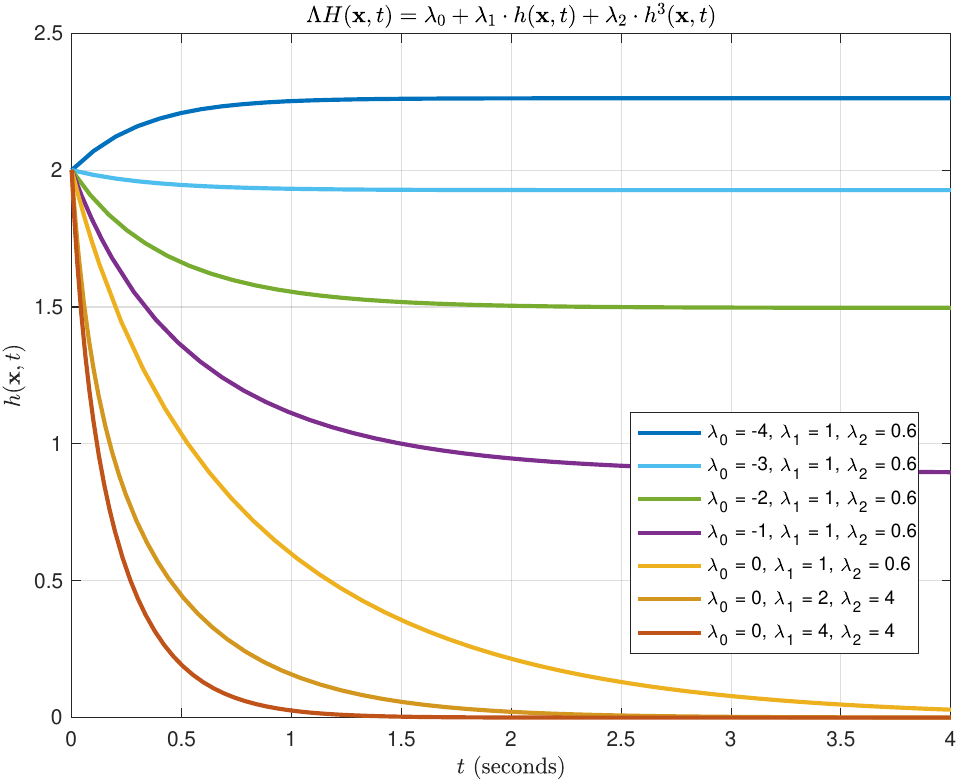}
    \caption{Arbitrary super-level set enforced by VL-CBF.}
    \label{fig.VL-CBF_graph}
\end{figure}

\subsection{Safety Guarantees for Two-Way Overtaking via VL-CBF}
\label{sec.safety_guarantees}
 For the two-way road overtaking problem, we need to in particular ensure that $veh_o$ does not collide with $veh_e$ while the ego vehicle stretching out from $\mathcal{L}_{\text{ego}}$ borrowing the $\mathcal{L}_{\text{opp}}$.
We show that, by equipping with the proposed VL-CBF, the ego vehicle and the opposing vehicle can keep safe as long as the initial condition is satisfied. 

Considering the longitudinal velocity of both vehicles, the safety condition for $veh_e$ is defined through the following value function
\begin{equation}
    h_{eo}(\mathbf{x}_e, t) \coloneqq x_o - x_e - \frac{(v_e^x - v_o^x)^2}{2 a_l}.
    \label{eq.two_way_value_function_e}
\end{equation}
For the opposing vehicle  $veh_o$, by assuming that it is also autonomous and adopts VL-CBF, the safety condition is defined as
\begin{equation}
    h_{oe}(\mathbf{x}_o, t) \coloneqq x_o - x_e - \frac{(v^x_e - v^x_o)^2}{2a_l}.
    \label{eq.two_way_value_function_o}
\end{equation}
In the above equations,  $a_l$ is the maximum deceleration rate, 
and $\mathbf{x}_e,\mathbf{x}_o$ stand for the states of $veh_e, veh_o$ respectively. 
We use $x_e$ to represent the first element of $\mathbf{x}_e$ in \eqref{eq.kinematic_bicycle} with $v^x_e\coloneqq \dot{x}_e$; 
similarly, $x_o$ and $v^x_o$ are two elements of $\mathbf{x}_o$ in \eqref{eq.double_integrator}. Furthermore, we introduce two symbols $\alpha_e^x, \alpha_o^x$ which will be used later, with $\alpha_e^x \coloneqq \ddot{x}_e$ and $\alpha_o^x$ is the control input of the system as defined  in Equation~\eqref{eq.double_integrator}. 
Note that in Equations~\eqref{eq.two_way_value_function_e} and \eqref{eq.two_way_value_function_o}, $x_o, v_o^x$ and $x_e, v_e^x$ are time-varying parameters respectively.
As   shown in Equations~\eqref{eq.two_way_value_function_e} and \eqref{eq.two_way_value_function_o}, the value functions $h_{eo}$ and $h_{oe}$ have the same expression but different independent variables, suggesting that they have the same interpretation of their relative status. 

The safety guarantee with respect to the opposing vehicle during the lane borrowing process is established as follows. 
\begin{Theorem}\upshape 

Consider VL-CBFs defined in Equations~\eqref{eq.two_way_value_function_e} and \eqref{eq.two_way_value_function_o}, and assume that 
both $veh_e$ and $veh_o$ follow the safety control law defined in Equation~\eqref{eq:vlcbf-input}. 
Then if the initial safety condition $\Lambda_\times H_\times \geq 0,~ \forall \times \in \{eo, oe\}$ is satisfied, the collision-free solution exists for both vehicles under the control limits.
    \label{thm.safety} 
\end{Theorem}

\begin{proof}
As the value function defined as \eqref{eq.two_way_value_function_e}\eqref{eq.two_way_value_function_o}, we have $H_{eo} = [1,h_{eo},\dots,h_{eo}^{2m-1}]$ and $H_{oe} = [1,h_{oe},\dots,h_{oe}^{2m-1}]$ sharing the same expression. For $veh_e$ in $\mathcal{L}_{\text{ego}}$, to satisfy VL-CBF we have
\begin{equation}
    \begin{aligned}
        \dot{h}_{eo} = v_o - v_e-\frac{2(v^x_e - v^x_o)(\alpha^x_e - \alpha^x_o)}{2a_l} \geq -\Lambda_{eo} H_{eo}
    \end{aligned}
    \label{eq.value_derivative}
\end{equation}
which yields
\begin{equation}
    \alpha^x_e - \alpha^x_o \leq \frac{a_l}{v^x_e - v^x_o}(\Lambda_{eo} H_{eo} + v^x_o - v^x_e), 
    \label{eq.VL-CBF_e}
\end{equation}
whereas for $veh_o$ in the opposite direction, we have 
\begin{equation}
    \alpha^x_e - \alpha^x_o \leq \frac{a_l}{v^x_e - v^x_o}(\Lambda_{oe} H_{oe} + v^x_o - v^x_e).
    \label{eq.VL-CBF_o}
\end{equation}
We first show that for all the scenarios under VL-CBF, we can find a joint relation so that the safety can be preserved under the maximum deceleration limit $a_l$. We define a new variable $\mathcal{H} \coloneqq \min\{\Lambda_\times H_\times\}, \times\in\{eo,oe\}$, and by combining \eqref{eq.VL-CBF_e} and \eqref{eq.VL-CBF_o}, we have
\begin{equation}
    \alpha^x_e - \alpha^x_o \leq \frac{a_l \mathcal{H}}{v^x_e - v^x_o} - a_l.
    \label{eq.VL-CBF_combined}
\end{equation}
To ensure the feasibility of the proposed barrier function, the right-hand side of \eqref{eq.VL-CBF_combined} must be greater than the lowest possible value of the difference in controller value when $\alpha^x_e = -a_l, \alpha^x_o = a_l$. Note that, $a_l$ is the absolute value of the maximum deceleration rate, which corresponds to vehicles decelerating at their extreme capacity. Thus, the following inequality must hold
\begin{equation}
    \frac{a_l \mathcal{H}}{v^x_e - v^x_o} \geq -a_l.
\end{equation}
Since $v^x_e - v^x_o > 0$ when $v^x_e$ is positive and $v^x_o$ is negative, the feasibility of this safety solution can be guaranteed if the two vehicles satisfy
\begin{equation}
    \mathcal{H} + v^x_e - v^x_o \geq 0.
    \label{eq.VL-CBF_feasibility}
\end{equation}
Satisfaction of \eqref{eq.VL-CBF_feasibility} suggests the existence of the solution at least at their extreme control capability.

Then, we show that \eqref{eq.VL-CBF_feasibility} holds if the initial condition is satisfied: we have $\mathcal{H} \geq 0$ as  $\Lambda_\times H_\times \geq 0$ is fulfilled. Since $v^x_e -v^x_o$ is positive, it follows that the inequality in \eqref{eq.VL-CBF_feasibility} is always satisfied, which indicates the feasibility of the original VL-CBF. According to Theorem~\ref{thm.VL-CBF_levelset}, safety can be ensured by the invariant property of its super-level set. This completes the proof.
\end{proof}
\smallskip

\begin{Remark}
In the above result, the initial safety condition $\Lambda_\times H_\times$ defines a nominal safe margin for the vehicle, indicating the edge of the safe space under the conservativeness degree $\epsilon$ decided by $\Lambda_\times$.
Note that, the system can further withstand disturbances by raising the VL-CBF level to increase conservativeness, and when $\lambda_0 = 0$, this initial condition simply indicates that under the maximum deceleration rate, the two vehicles can stop at the very last moment with the minimum conservative distance they defined, which is their physical limit at the most extreme situation.
\end{Remark}

\section{Two-Way Road Overtaking Framework}
\label{sec.framework}
In the preceding section, we have developed the basic tool of VL-CBF for ensuring safety in the presence of opposing vehicles. 
In this section, we present our overall overtaking framework, which can tackle both the front vehicle as well as the potential opposing vehicle. To present the integrated framework, we proceed in this section as follows:
\begin{itemize}[leftmargin=*] 
    \item 
    First, we focus on the problem of overtaking the front vehicle without considering the potential opposing vehicle. 
    Our approach is based on solving a time-optimal control problem by MPC in which the previously introduced VL-CBF is used for safety guarantees. 
    \item 
    Then we introduce an integrated overtaking framework for dealing with the potential opposing traffic. Specifically, at each instant, we need to solve two safety-guarantee optimization problems: one for overtaking the front vehicle and the other for returning to the original position. 
    We prove that at least one of these two problems is always feasible, which provides a formal guarantees of safety for the overall framework. 
    \item 
    Finally, we discuss the extensibility of our framework to further incorporate the  semi-autonomous surrounding vehicles with human drivers. We show how safety guarantees can  still be ensured through conservative analysis.
\end{itemize} 


\subsection{CBF-Based Time-Optimal Model Predictive Control}
First, we focus on the overtaking scenario only involving the front vehicle. 
We formulate this scenario as a time-optimal control problem with safety constraints. 
Specifically, we solve this problem using model predictive control (MPC) whose objective function is to minimize the overtaking time yet maintaining a safe distance from the front vehicle, together with other constraints to consider the vehicle dynamics and performance factors.

\begin{figure} 
    \centering
    \includegraphics[width = 0.75\linewidth]{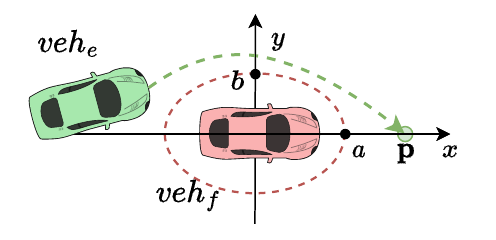}
    \caption{Relative ellipse around the front overtaken vehicle.}
    \label{fig.relative_CBF}
\end{figure} 

In order to efficiently handle the safety requirement with respect to the front overtaken vehicle, VL-CBF is used as a constraint in the optimization problem.  
As shown in Figure~\ref{fig.relative_CBF}, we model the safety requirement with respect to the front vehicle by an ellipse that covers the front vehicle, where $a, b$ are the radius  of the ellipse.  
Let $x_f$ and $y_f$ be the longitudinal and lateral states of the front overtaken vehicle, respectively, and we denote  $\mathbf{x}_f=[x_f,y_f]$. 
$\mathbf{x}_e$ is used as Section~\ref{sec.VL-CBF} to denote the state of the ego vehicle.
We define a (VL-)CBF 
\begin{equation}
    h_{ef}(\mathbf{x}_e,t) \coloneqq \beta_1\|x_e(t) - x_f(t)\|^2 + \beta_2\|y_e(t) - y_f(t)\|^2 - \beta_3,
\end{equation}
where $\beta_i,~i\in\{1,2,3\}$ are hyper-parameters satisfying $\beta_1/\beta_3 = a^2$ and $ \beta_2/\beta_3 = b^2$.  
Then this CBF defines the desired ellipse safe region  $\mathcal{C}^f$ with respect to the overtaken vehicle $veh_f$, i.e., 
\begin{equation}
    \mathcal{C}^f(t) \coloneqq \{\mathbf{x}_e\in\mathbb{R}^4: h_{ef}(\mathbf{x}_e,t) \geq 0\}.
    \label{eq.nominal_ellipse}
\end{equation} 
By incorporating the VL-CBF, the safety constraint on the control input $\mathbf{u}_e$ for overtaking $veh_f$ becomes
\begin{equation}
    L_f h_{ef} + L_g h_{ef}\mathbf{u}_e + \frac{\partial h_{ef}}{\partial t} \geq \Lambda_{ef}H_{ef},
\end{equation}
which resulting in a non-negative super-level set
\begin{equation}
    \mathcal{C}_\epsilon^f(t) \coloneqq \{\mathbf{x}_e\in\mathbb{R}^4: h(\mathbf{x}_e,t) \geq{\epsilon}\},
\end{equation}
with $\epsilon \in\mathbb{R}_{\geq{0}}$.

Now we formally introduce the time-optimal CBF-based model predictive control problem as described in the optimization problem TO-CBF-MPC.
Let $N$ be a given prediction horizon of the entire \emph{steps} of the optimization problem.   Then the problem is formulated as follows:
\begin{itemize}[leftmargin=*] 
    \item 
    The decision variables consist of a trajectory $\{\mathbf{x}_e\}^N$ of the ego vehicle, a control input sequence $\{\mathbf{u}\}_e^N$ and a sequence of time intervals $\{\Delta t\}^N$. 
    We denote by $\mathbf{x}_e^i$, $\mathbf{u}_e^i$ and $\Delta t_i$ the $i$th entries in  $\{\mathbf{x}_e\}^N$, 
    $\{\mathbf{u}_e\}^N$ and $\{\Delta t\}^N$ respectively. 
    Then  $\Delta t_i$ is the time interval between states $\mathbf{x}_e^i$ and $\mathbf{x}_e^{i+1}$ 
    within which control input $\mathbf{u}_e^i$ is applied.  
    Their relations are captured by the   constraints in Equations~\eqref{eq.tompc_dy}, \eqref{eq.tompc_u} and \eqref{eq.tompc_v}. 
    Furthermore, Equations~\eqref{eq.tompc_t} put constraint on the maximum time interval between each intermediate point in order to mitigate the inter-sampling effect \cite{xiaoDifferentiableControlBarrier2022}.
    \item 
    The terminal constraint of the trajectory is captured by Equations~\eqref{eq.tompc_goal} and \eqref{eq.tompc_y}. 
    Specifically, we choose a reference goal location $\mathbf{p} = [x_g,y_g]$ in front of the overtaken vehicle in its longitudinal direction. Parameter $\epsilon_y$ is the lateral bias parameter to limit the optimized final position.  
    Here $x_g$ is usually selected  based on the speed and reaction time of the front vehicle $veh_f$ with $x_g = x_f + \phi v_f(t)$, $\phi$ is the reaction time (usually takes 1.8s \cite{xiaoDifferentiableControlBarrier2022}).
    \item 
    The safety constraint for the front vehicle (or surrounding vehicle during the overtaking process) is captured by Equation~\eqref{eq.tompc_cbf}, which is a VL-CBF constraint that ensures all planned states $\{\mathbf{x}\}^N$ are collision-free with minimum clearance to the ellipse defined in \eqref{eq.nominal_ellipse} of $\tilde{\kappa}^{-1}(\lambda_0)$, where $\tilde{\kappa}^{-1}$ is the inverse function of $\tilde{\kappa}$ in Theorem~\ref{thm.VL-CBF_levelset}, and $\lambda_0$ here is the first element of $\Lambda_{ef}$ in \eqref{eq.tompc_cbf}.
    \item 
    The overall optimization objective is to estimate and minimize the total time taken before reaching the goal location, i.e., $\sum_{i=0}^{N-1} \Delta t_i$. Therefore, the result of the optimization problem also provides us an \emph{estimated time} for completing the overtaking task. 
\end{itemize}

\smallskip

\begin{boxD}
{ \textbf{TO-CBF-MPC:}}
\begin{subequations}
    \begin{align}
        &\mathop{\min}_{\{\mathbf{u}_e\}^N\in\mathcal{U}^N,\{\mathbf{x}_e\}^N\in\mathcal{X}^N,\{\Delta t\}^N} J = \sum_{i=0}^{N-1} \Delta t_i\\
         &\text{subject to}~ \nonumber \\
         \begin{split}
             & \mathbf{x}_e^{i+1} = f(\mathbf{x}_e^i) + g(\mathbf{x}_e^i)\mathbf{u}_e^i\Delta t_i, \\
             & i\in \{0,\dots,N-1\}\label{eq.tompc_dy}
         \end{split}
         \\
         & \mathbf{u}_e^i\in\mathcal{U},~i\in \{0,\dots,N-1\} \label{eq.tompc_u}\\
         & \mathbf{x}_e^i\in\mathcal{X},~i\in \{0,\dots,N-1\}\label{eq.tompc_v}\\
         & \Delta t_i\leq t_{max},~i\in \{0,\dots,N-1\}\label{eq.tompc_t}\\
         & x_e^N \geq x_g \label{eq.tompc_goal}\\
         & y_e^N \in [y_g - \epsilon_y, y_g + \epsilon_y ], ~ \epsilon_y > 0 \label{eq.tompc_y}\\
         \begin{split}
             & L_f h_{ef}^t+ L_g h_{ef}^t\mathbf{u}_e^i +\frac{\partial h_{ef}^t}{\partial t} \geq -\Lambda_{ef} H_{ef}^t,\\
             & i\in \{0,\dots,N-1\} \label{eq.tompc_cbf}
         \end{split}
    \end{align}
    \label{eq.tompc}
\end{subequations}
\end{boxD}


Note that when applying MPC-based control, the above optimization problem needs to be solved iteratively based on the current actual state. Specifically, for each time instant, 
only the  first control entry $\mathbf{u}_e^0\in\{\mathbf{u}_e\}^N$ in the computed input sequence is applied. 
The control input will be updated when the next optimization problem is solved.
Then at the start of next iteration, we re-measure the states of $veh_e$ and $veh_f$, and re-compute $\{\mathbf{u}_e\}^N$ by solving TO-CBF-MPC,   still apply its control entry, and so forth. The convergence criteria is met when the first state $\mathbf{x}_e^0$ at the next iteration satisfies the relationships $x_e^0\geq{x_g}$ and $y_e^0\in[y_g - \epsilon_y, y_g + \epsilon_y]$, which are analogue to Equations~\eqref{eq.tompc_goal} and \eqref{eq.tompc_y}.


\begin{Remark}
It is worth remarking that the safe-by-construction property of our solution is ensured by the VL-CBF constraint
and is independent from the objective function in the optimization problem. 
In the proposed TO-CBF-MPC approach, we only consider the overtaking time as the single objective function in order to simplify the setting. However, this is without loss of generality since one can add any additional  user-preferred performance metric in the objective function, which will not affect the overall safety property. 
    For example, if one wants to further limit the changing rate of the acceleration to smooth the driving behavior, then  $\frac{\gamma_{\mathbf{u}}}{N}\cdot\sum||\mathbf{u}_e^{i+1} - \mathbf{u}_e^i||,~i\in\{0,\cdots,N-1\}$ can be added to the objective function with $\gamma_{\mathbf{u}}\geq 0$ as a user-defined weight \cite{althoffProvablycorrectComfortableAdaptive2021}.  
\end{Remark}

\subsection{Integrated Overall Overtaking Framework}
\label{sec.integrated_framework}

With the previously introduced TO-CBF-MPC for motion planning and overtaking time estimation, we now consider the scenario when opposing traffic occurs while we are still borrowing the opposite lane. 
In this case, we need to consider the possible position of $veh_o$ within the time horizon $\{x_o\}^N$ to safely plan our motion via the above TO-CBF-MPC.

Still, we assume that  $veh_o$ is autonomous and follows a (possibly non-identical) VL-CBF control law whose parameters $\Lambda_o$ can be obtained by the ego vehicle $veh_o$ via, for example,  communications.
Then by knowing parameters $\Lambda_o$, the ego vehicle can precisely estimate the control upper bound of the opposing vehicle when driving towards it. 
Therefore, when solving the TO-CBF-MPC problem, the following constraints need to be added into \eqref{eq.tompc} in order to ensure safety with respect to the opposing vehicle

\begin{subequations}
    \begin{align}
        & \mathbf{x}_o^{i+1} = f_o(\mathbf{x}_o^i) + g_o(\mathbf{x}_o^i) \mathbf{u}_o^i \Delta t_i \label{eq.tompcAdd_dynamic_o}\\
        & L_f h_{eo}^t + L_g h_{eo}^t \mathbf{u}_e^i + \frac{\partial h_{eo}^t}{\partial t}\geq -\Lambda_{eo}H^t_{eo} \label{eq.tompcAdd_cbf_eo}\\
        & L_{f_o} h_{oe}^t + L_{g_o} h_{oe}^t \mathbf{u}_o^i + \frac{\partial h_{oe}^t}{\partial t}\geq - \Lambda_{oe}H^t_{eo} \label{eq.tompcAdd_cbf_oe},
    \end{align}
    \label{eq.tompcAdd}
\end{subequations}
\vspace{-6pt}

\noindent for $i\in\{0,\dots,N-1\}$, where \eqref{eq.tompcAdd_dynamic_o} is the dynamic for $veh_o$ in \eqref{eq.double_integrator}. Then, after incorporating the new VL-CBF constraints above, we can jointly get the estimated sequence of $\{\mathbf{x}_e\}^N, \{\mathbf{u}_e\}^N, \{\mathbf{x}_o\}^N, \{\mathbf{u}_o\}^N$ within the prediction horizon, as well as the sequence of intervals $\{\Delta_t\}^N$ between time steps.
Since VL-CBF \eqref{eq.tompcAdd_cbf_oe} regulates the most radical behavior possible of $veh_o$, by solving the re-formulated TO-CBF-MPC problem, the resulting control action is collision-free to both $veh_f$ and $veh_o$ by construction, which is formally stated by the following result.



\begin{Theorem}\upshape
	\label{thm.safe-by-construction}    
 	Suppose $veh_f$ is non-accelerating during the overtaking process, the proposed overtaking control strategy is safe for the ego vehicle as long as the initial feasibility of the incorporated TO-CBF-MPC problem is satisfied.
\end{Theorem}
\proof
The initial feasibility of the above optimization problem ensures the solution exists under current velocity of $veh_f$, and the potential autonomous opposing traffic, and therefore, the safety is always preserved under the longitudinal and ellipse-shaped VL-CBFs.
And since $veh_f$ is non-accelerating, the actual time taken to overtake the $veh_f$ ($\sum\Delta\overline{t}_i$) will be less than or equal to the planned overtaking time $\sum \Delta t_i$, as we consider the relative position and velocity w.r.t. $veh_e$. Thus, the overtaking can be completed with less time taken, suggesting that the constraint in \eqref{eq.tompc_t} will always be satisfied, and the system can accept smaller control input $\mathbf{u}$ that being contained in the original solution space. This completes the proof.
\endproof

\begin{Remark}
The developments of our framework are based on the   assumption that the strategy of the opposing vehicle is known to the ego vehicle. 
When this assumption does not hold, our framework can still be applied by further leveraging some techniques from the literature. 
One possible approach  is to  utilize the   \emph{parameter learning by interaction} technique in [43] that learns the   driving parameters $\Lambda_o$ through interactions on-the-fly.  Another approach  is to use the \emph{conservative analysis} technique  to estimate the worst-case of  $veh_o$ during the process. 
This approach is  more general as it can be even applied to any possible control strategy for $veh_o$, but consequently, it also is more conservative.  Details on how this conservative analysis can be integrated into our framework is provided in Section~V.C.   
\end{Remark}

\begin{figure*}
    \newcommand{\colwidth}{3.8cm}
	\renewcommand{\tabcolsep}{0.4mm}
    \centering
    \includegraphics[width = \linewidth]{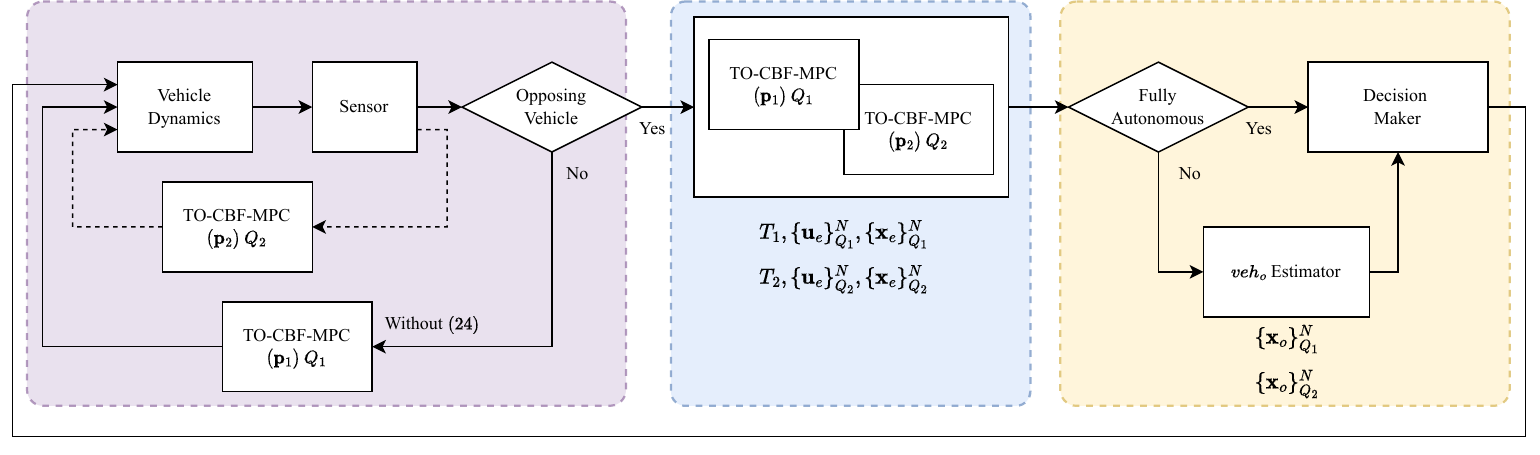}
    \caption{The two-way road overtaking framework considering  opposing traffic. The framework   consists of three parts (emphasized in the dashed color blocks), 
    where  
    the left purple block depicts the control diagram when no opposing vehicle is around, 
    the middle blue block depicts the \emph{dual} TO-CBF-MPC mechanism for forward and backward planning, 
    and the right yellow block depicts the switch to our decision maker for coping with human driver-based opposing vehicles.}
    \label{fig.overall_framework}
\end{figure*}


The formal safety guarantee in Theorem~\ref{thm.safe-by-construction} is based on the assumption that $veh_f$ is cooperative in the sense that it is non-accelerating during the overtaking process. In practice, if $veh_f$ is also accelerating, then overtaking task may not be feasible. In this scenario, the ego vehicle needs a mechanism to recognize this issue and to take defensive actions such as stopping the overtaking process. 

To address the above issue, we introduce a \emph{dual TO-CBF-MPC} structure which provides ego vehicle the option to back to $\mathcal{L}_{\text{ego}}$ when overtaking is not feasible.  Our integrated two-way road overtaking strategy with dual TO-CBF-MPC structure is illustrated in Figure~\ref{fig.overtaking_diagram2}. Specifically, instead of solving a safe overtaking control problem, we further solve a safe control problem, where the goal is to return back to the front vehicle. Therefore,  two different reference goal locations $\mathbf{p}_1 = [x_g^1,y_g^1]$ and $\mathbf{p}_2 = [x_g^2,y_g^2]$ are  selected, one in front of $veh_f$ and the other behind $veh_f$, which correspond to two TO-CBF-MPC problems (with different $[x_g,y_g]$); we term the MPC problem from $[x_e,y_e]$ to $[x_g^1, y_g^1]$ as $Q_1$, and to $[x_g^2,y_g^2]$ as $Q_2$. To put it simply, $Q_1$ plans the actions to overtake $veh_f$, while $Q_2$ serves as a back-up strategy for this safety-critical system to abandon the overtaking process when the feasibility of overtaking is at risk.

\begin{figure} 
    \centering
    \includegraphics[width = 0.9\linewidth]{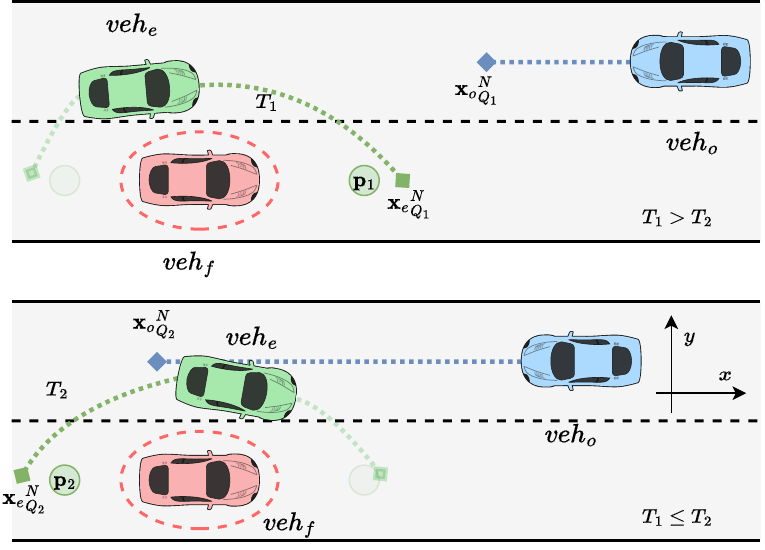}
    \caption{Illustration of the two-way road overtaking strategy. The upper figure shows the prediction process of the TO-CBF-MPC problem $Q_1$, where the green dotted line represents the estimated $N$ points of $veh_e$'s trajectory within the time $T_1$ to reach $\mathbf{p}_1$ and the blue dotted line represents the estimated $veh_o$'s trajectory during this time. 
    The lower figure shows the process of solving TO-CBF-MPC problem $Q_2$ when $Q_1$ is infeasible.}
    \label{fig.overtaking_diagram2}
\end{figure}

At the beginning of the overtaking maneuver, the ego vehicle $veh_e$ detects initial feasibility of $Q_1$, i.e., whether there is potential conflict when placing the reference goal location $\mathbf{p}_1$, and whether there exist opposing traffic that may cause unexpected danger (violates the initial conditions in Theorem~\ref{thm.safety}). If the initial feasibility condition is fulfilled, then the overtaking maneuver begins. The proposed overtaking framework is illustrated in Figure~\ref{fig.overall_framework}. It starts by continuously detecting whether there exists opposing vehicle or not. If no opposing vehicle occurs, then it iterates through $Q_1$ by applying $\mathbf{u}^0$ at each iteration until convergence. After moving out of the current lane, $Q_2$ becomes active, with $\mathbf{p}_2$  as reference goal location relative to $veh_f$ at its back. When $Q_1$ is unsolvable, e.g., due to the acceleration of $veh_f$, the system shifts toward the loop of $Q_2$, shown in the dashed-loop in Figure~\ref{fig.overall_framework}, and  it will be guided back into its original lane. 
The following theorem shows that our framework ensures at least one MPC problem has a solution at each instant.

\begin{Theorem}\upshape
    \label{thm.one_in_two_feasibility} 
    By adapting the integrated overtaking framework with the dual TO-CBF-MPC structure shown in  Figure~\ref{fig.overall_framework},  at least one of the aforementioned TO-CBF-MPC problems $Q_1,Q_2$ is feasible during the overtaking process  
    regardless of $veh_f$ is accelerating or decelerating, but not interchangeably. 
\end{Theorem} 
\proof
The start of the overtaking process suggests the satisfaction of the initial feasibility, which states that there exists solution to problem $Q_1$ and $Q_2$ so that $\mathbf{x}_e$ to $\mathbf{p}_1$ and $\mathbf{p}_2$ are feasible. For overtaking problem $Q_1$, the in-feasibility of \eqref{eq.tompc} can only be affected by (\ref{eq.tompc_t}), i.e., when $T_1 > N\cdot t_{max}$. This corresponds to the case that $veh_f$ accelerates when $veh_e$ overtaking, which means the relative position and velocity of $veh_f$ is beyond our initial measurements. We denote the predicted longitudinal distance between $\mathbf{x}_e$ and $\mathbf{p}_1,\mathbf{p}_2$ as $l_x^1,l_x^2$ respectively, then the violation of (\ref{eq.tompc_t}) suggests that $\|x_e - x_g^1\|>l_x^1$.  Since $\mathbf{p}_1,\mathbf{p}_2$ is placed relative to the $\mathbf{x}_f$, it suggests $\|x_g^1- x_g^2\| = c$, where $c$ is a constant. Thus, for $\mathbf{x}_e$ to $\mathbf{p}_2$ we have $\|x_e - x_g^2\|<l_x^2$, resulting in a smaller relative distance in longitudinal direction; and since $veh_f$ is accelerating, the relative velocity (acceleration) between $veh_e$ and $veh_f$ is also greater for $Q_2$, therefore, with shorter distance and larger acceleration gap, $T_2 < \hat{T}_2$, where $\hat{T}_2$ is the initially estimated time for $Q_2$. Since the satisfaction of initial feasibility suggests $\hat{T}_2\leq N\cdot t_{max}$, which leads to $T_2\leq N\cdot t_{max}$ that guarantees the feasibility of $Q_2$. Conversely, if $T_1 \leq N\cdot t_{max}$, solution from $\mathbf{x}_e$ to $\mathbf{p}_1$ always exists. This completes the proof.
\endproof

The feasibility of $Q_1$ suggests that the path and control sequence planned at the current position can complete the overtaking action without collision. Whereas the feasibility of $Q_2$ suggests that there exist path and control sequence such that the ego vehicle can return back to its original lane without collision. When the $veh_f$ is non-accelerating, we can ensure that $Q_1$ is feasible, and the overtaking action can be completed; when the feasibility of $Q_1$ is challenged due to the acceleration of $veh_f$, we have $Q_2$ as our formal safety solution that can lead the ego vehicle back to its original lane safely.

\begin{Remark}
The optimization problem $Q_2$ in the proposed dual TO-CBF-MPC framework also serves as a fault-tolerant mechanism that further ensures the safety of the synthesized control strategy under non-cooperative environments. 
That is, under the extended possible accelerating behavior of $veh_f$ making $Q_1$ unsolvable, our framework can enable the ego vehicle to turn back to its original position in $\mathcal{L}_{\text{ego}}$ safely. 
 Overall,  Theorems~\ref{thm.safe-by-construction} and \ref{thm.one_in_two_feasibility}  show that our integrated framework can provide formal guarantees for general on-road overtaking problems under the potential opposing traffic. 
\end{Remark}

\begin{Remark}
Note that,  the assumption that $veh_f$ is accelerating or decelerating but not interchangeably is essential for the correctness of our method. Without this assumption, in fact, no safe control law exists. For example,  suppose that, when the ego vehicle is trying to overtake the front vehicle in the adjacent lane,  the front vehicle takes aggressive actions by exactly imitating the behavior of the ego vehicle. That is, it accelerates as the ego vehicle accelerates, and decelerates as the ego vehicle decelerates.  Then in this case, the  ego vehicle can neither overtake nor merge back to the safe region.  However, this scenario is usually considered as ``non-rational" or even illegal in our daily life.  Particularly,  as pointed out by the well-known Responsibility-Sensitive Safety (RSS) model in autonomous driving  \cite{shalev-shwartzFormalModelSafe2017}, one vehicle must obey the RSS rules in order to not be the ``cause'' of an accident,  and this scenario is essentially a violation of the RSS rules. 
\end{Remark}


\subsection{Incorporating  Semi-Autonomous Vehicles}

In Section~\ref{sec.integrated_framework}, we have established a formal safety guarantee for the proposed dual TO-CBF-MPC framework under the assumption that the 
opposing vehicle is also autonomous. 
In this subsection, we further discuss how the formal safety guarantee can be carried to the case when $veh_o$ is semi-autonomous (non-cooperative) by conservative analysis.


Since the behaviors of human driver are usually irregular, we set the maximum acceleration and speed limit to the semi-autonomous $veh_o$ to get the over-approximated prediction horizon $\{\mathbf{x}_o\}^N$ based on our real-time measurement of its states at every time step. Then the decision making process are generally depending on three key status: time taken to complete $Q_1,Q_2$; the final planned state of $veh_e$, ${x_e}_{Q_i}^N,~i\in\{1,2\}$; and the final predicted state of $veh_o$, ${x_o}_{Q_i}^N,~i\in\{1,2\}$;
where $\{\times\}_{Q_i}^N, i\in\{1,2\}$ represents the predicted sequence of $N$ corresponding states (or controls) in the curly brackets by solving the TO-CBF-MPC problem of $Q_i$; $\times_{Q_i}^k$ denotes the $k^{th}$ entry of the corresponding sequence.

As illustrated in Figure~\ref{fig.overtaking_diagram2}, when $T_1 > T_2$, we only care about the safety of the final state while overtaking, since the time takes to merge back to $\mathbf{p}_2$ is smaller than that of overtaking to $\mathbf{p}_1$, we can always merge back, as long as we can overtake. When $T_1 \leq T_2$, we have to consider whether there exists solution to $Q_2$ if overtaking become infeasible (can be caused by the approaching behavior of $veh_o$). Since merging back is not time-critical, the priority of safety control becomes ensuring the feasibility of $Q_2$; thus, when the feasibility of $Q_2$ is at risk, i.e., ${x_o}_{Q_2}^N$ is approaching $\mathbf{p}_2$, we abandon the overtaking process and apply the control planned by $Q_2$. Note that, if the feasibility of $Q_2$ is at risk when $veh_e$ already cut into $\mathcal{L}_{\text{ego}}$ in front of $veh_f$, we will keep finishing the overtaking process. However, under all cases, whenever the feasibility of $Q_1$ is being challenged, we merge back. Our decision maker is briefly summarized in Table \ref{tab.cons_for_overtaking}.
\begin{table} [!ht]
    \centering
    \arrayrulecolor{black}
    \caption{Conditions for Overtaking}
    \label{tab.cons_for_overtaking}
    \begin{adjustbox}{width = 1\linewidth}
    \begin{tabular}{|cccc|c|c|}
        \hline
        \multicolumn{4}{|c|}{Conditions}&\multirow{2}{*}{Overtake}&\multirow{2}{*}{Go Back}\\
        \cline{2-3}
        $T_1>T_2$&{\tiny ${x_o}_{Q_1}^N > {x_e}_{Q_1}^N$} & {\tiny ${x_o}_{Q_2}^N > x_g^2$} &$y_e\in\mathcal{L}_{\text{ego}}$& & \\
        \hline 
        $\bullet $&$\bullet $& & &$\surd $&\\
        \hline
        $\circ $&$\bullet $& &$\bullet $&$\surd $&\\
        \hline
        $\circ $&$\bullet $& &$\circ $& &$\surd $\\
        \hline
         &$\bullet$& $\bullet$ & &$\surd$ & \\ 
        \hline
        &$\circ $& & & &$\surd $\\
        \hline
    \end{tabular}
    \end{adjustbox}
    $\bullet$: satisfies, $\circ$: violates, $else$: not care.
\end{table}

In Table \ref{tab.cons_for_overtaking}, the first column states the condition that whether the estimated time for TO-CBF-MPC problem $Q_1$ is greater than that of $Q_2$; the second column represents whether the predicted final position of $veh_o$ will collide with the planned final position of $veh_f$; the third column checks whether the merging back operation will be safe when $T_1\leq T_2$; and the fourth column checks if the ego vehicle is already in $\mathcal{L}_{\text{ego}}$ at the moment. The details of our framework is presented in Algorithm \ref{alg.framework}.

\begin{algorithm} 
    \caption{Two-way Road Overtaking Framework}
    \label{alg.framework}
    \KwIn{$N, \mathbf{p}_1, \mathbf{p}_2$}
    \KwOut{$T_1,T_2, \{\mathbf{u}_e\}_{Q_1}^N, \{\mathbf{u}_e\}_{Q_2}^N$}
    overtaking = True\\
    $veh_o$ = False\\
    
    \Do{$Q_1$ \textbf{not} converged \textbf{and} overtaking == 
    True}{
        get relative position of $veh_f$ w.r.t. $veh_e$\\
        set virtual points $\mathbf{p}_1$ and $\mathbf{p}_2$ according to $veh_f$\\
        detect $veh_o$\\
        \If{$veh_o$ == True}{
            get measurements $\mathbf{x}_o$\\
            Add VL-CBF constraint \eqref{eq.tompcAdd}\\
            compute $T_1, \{\mathbf{x}_e\}_{Q_1}^N, \{\mathbf{u}_e\}_{Q_1}^N$ via $Q_1$\\
            compute $T_2, \{\mathbf{x}_e\}_{Q_2}^N, \{\mathbf{u}_e\}_{Q_2}^N$ via $Q_2$\\
            \If{$veh_o$ Semi-autonomous}{
                Estimate worst case $\{\mathbf{x}_o\}_{Q_1}^N$, $\{\mathbf{x}_o\}_{Q_2}^N$\\
                overtaking $\leftarrow$ decision from Table \ref{tab.cons_for_overtaking}\\
                \If{overtaking == False}{
                    break\\
                }
            }
            \Else{
                \If{$Q_1$ infeasible}{
                    break\\
                }
            }
        }
        apply ${\mathbf{u}_e}_{Q_1}^1$ via $Q_1$\\
        
    }
    \Do{ $Q_2$ \textbf{not} converged \textbf{and} overtaking == False}{
        apply ${\mathbf{u}_e}_{Q_2}^1$ via $Q_2$\\
    }
\end{algorithm}

\smallskip
\begin{Remark}
    As the above section suggests, our framework can also be applied to the scenario when the opposing vehicle is semi-autonomous or human-driver-based. In this case, we use the conservative analysis by assuming the maximum acceleration rate and the speed limit of the $veh_o$ in the prediction horizon and using the decision table to aid the decision-making process. Although this setting will increase the conservativeness under a semi-autonomous scenario, it provides formal guarantees of safety under the worst case. 
\end{Remark}


\section{Simulation Results}
\label{sec.sims}

Having illustrated the proposed two-way road overtaking framework in the previous sections, we now demonstrate our algorithm through simulations. Note that, the color convention for vehicles in the simulations are inherited from Section~\ref{sec.problem}.

\subsection{Simulation Setups}
\label{sec.sim_setups}
In our simulations, the dimension and physical parameters of all vehicles involved are adapted from Toyota Camry, with length of 4.885m and width of 1.840m. The lane width is scaled according to averaged global standard as 3.5m. Horizon length of the predicted steps is $N = 50$, and maximum time resolution is $10^{-4}s$. The sampling rate towards the $veh_o$ is $50Hz$. 
To further demonstrate the robustness of our method against uncertainties, random disturbances $\boldsymbol{\mu} = [\mu_x,\mu_y,\mu_v]\in\mathbb{R}^3$ are injected into the perception module, i.e., the state information used in the optimization problem are the actual states with noises.  
Specifically, the perception noises are of uniform distribution such that: for the position information, we have 
$\mu_x,\mu_y\in[-0.5,0.5]$ when the distance with the measured objects is greater than 2m, and $\mu_x,\mu_y\in[-0.1,0.1]$ otherwise; and for velocity information, we have  $\mu_v$ in the $\pm 10\%$ range of the actual velocity.


The TO-CBF-MPC optimization problem   in  Equation~\eqref{eq.tompc} is implemented in \texttt{Python} with \texttt{CasADi} \cite{anderssonCasADiSoftwareFramework2019} as modeling language, and  is  solved via IPOPT  method \cite{bieglerLargescaleNonlinearProgramming2009} using a single performance core of the M1 Pro ARM processor.

\subsection{Overtaking with No Opposing Vehicle}

As illustrated in Figure~\ref{fig.sims1}, when there is no opposing traffic, the overall problem becomes a one-way road overtaking problem. During the overtaking process, the front vehicle $veh_f$ can vary its speed, but it will adopt the rational driving behavior as stated in Section~\ref{sec.problem}. The reference goal positions $\mathbf{p}_1$ and $\mathbf{p}_2$ are displayed as the faded green rectangles in front of and behind the $veh_f$. The dashed ellipse around $veh_f$ is the safety boundary enforced by the CBF based on the relative position between $veh_e$ and $veh_f$. 
\begin{figure}[!ht]
    \centering
    \includegraphics[width = 1\linewidth]{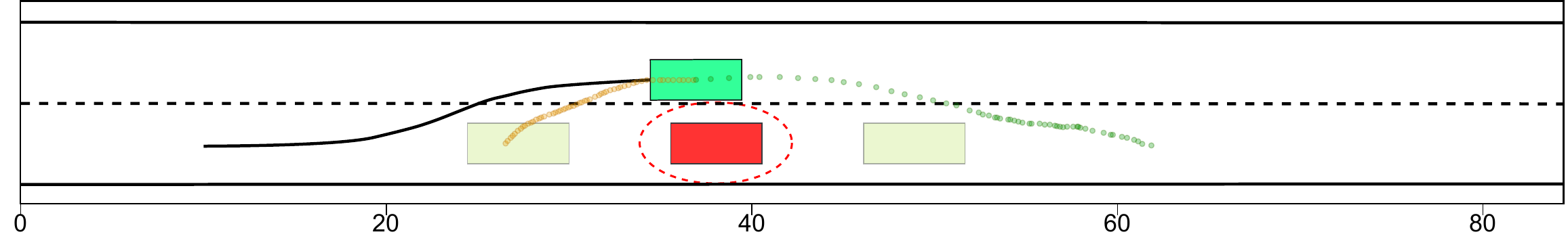}
    \caption{Overtaking with no opposing vehicle with constant velocity of the front vehicle at 8.3m/s (30km/h), and the velocity of the ego vehicle  ranges from 10-19.4m/s (36-70km/h).}
    \label{fig.sims1}
\end{figure}

 In the simulation figures, the green and orange dotted circular trajectories represent the planned trajectory from $\mathbf{x}_e$ to $\mathbf{p}_1$ and $\mathbf{p}_2$ respectively derived through $Q_1$ and $Q_2$. The solid black line at the back of $veh_e$ shows the actual trajectory traveled by the ego vehicle.

\begin{table}[!ht]
    \centering
    \arrayrulecolor{black}
    \caption{Time taken when no opposing traffic}
    \label{tab.sim1_time}
    \begin{tabular}{|c|>{\color{black}}c|>{\color{black}}c|>{\color{black}}c|>{\color{black}}c|}
        \hline
         speed range of $veh_e$ & \multicolumn{4}{c}{\color{black}36-70km/h}\vline\\
        \hline
         {\color{black}speed of $veh_f$}& 25km/h & 28km/h & 32km/h & 36km/h\\
         \hline
         \hline
         min (s)& 4.13 & 4.83 & 4.99 & 5.29\\
         \hline
         max (s)& 5.27 & 5.94 & 6.44 &6.63\\
         \hline
         mean (s)& 4.64 & 5.02 & 5.59 &5.90\\
        \hline
    \end{tabular}
\end{table}

The time taken to overtake the front vehicle with different velocity configurations $v_f$ are summarized in Table~\ref{tab.sim1_time}. For each group, 30 cases were simulated. The ego vehicle starts from the position at 10m at 10m/s (36km/h), with an upper speed limit at 19.4m/s (70km/h). The front vehicle starts at 64m. Note that, to cope with the introduced disturbances, the level set in the VL-CBF was set at $\epsilon = \tilde{\kappa}^{-1}(-\lambda_0) = 0.3$.

\subsection{Overtaking with Opposing Vehicles}

For the case with opposing vehicle during the overtaking process, the following different scenarios are simulated to demonstrate the effectiveness of our overall framework:
    \begin{itemize}[leftmargin=*] 
        \item $veh_f$ is non-accelerating and $veh_o$ is autonomous with a known VL-CBF strategy; 
        \item $veh_f$ is non-accelerating and $veh_o$ is human driver-based  with an unknown strategy; 
        \item $veh_f$ is accelerating  and $veh_o$ is autonomous with a known VL-CBF strategy.
    \end{itemize}
    All simulations are conducted with prediction horizon $N = 50$ and with perception uncertainties introduced in Section~\ref{sec.sim_setups}.

\begin{figure}[!ht]
    \centering
    \includegraphics[width = 1\linewidth]{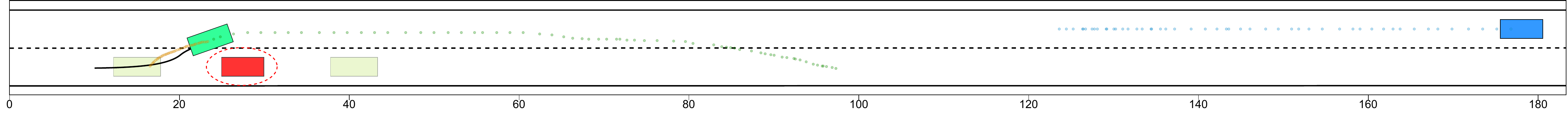}\medskip\\
    \includegraphics[width = 1\linewidth]{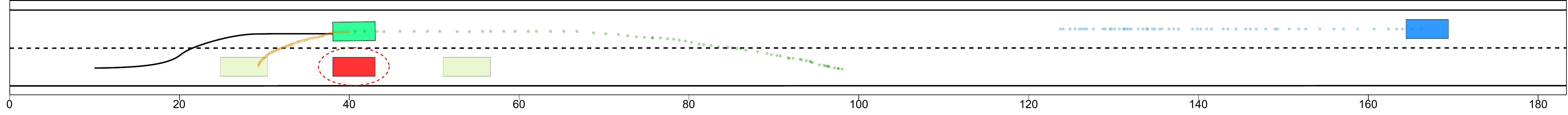}\medskip\\
    \includegraphics[width = 1\linewidth]{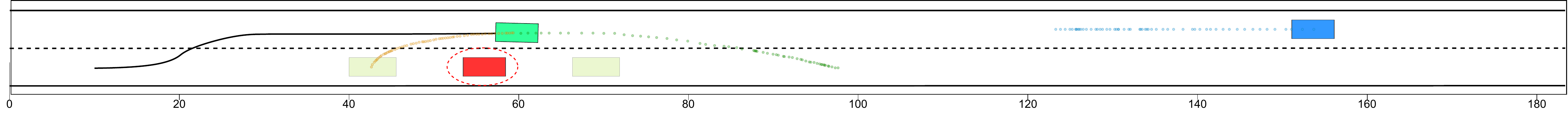}\medskip\\
    \includegraphics[width = 1\linewidth]{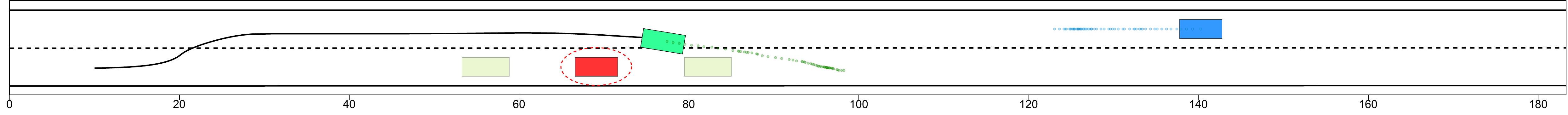}\medskip\\
    \includegraphics[width = 1\linewidth]{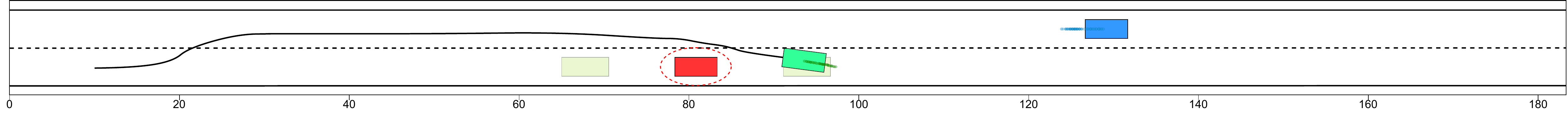}
    \caption{Snapshots from the simulation of the overtaking framework, with $veh_f$ non-accelerating and $veh_o$ fully autonomous.}
    \label{fig.sims2_1}
\end{figure}

Case 1:  $veh_f$ is non-accelerating and $veh_o$ is autonomous with a known VL-CBF strategy.

For this case, Figure~\ref{fig.sims2_1} demonstrates a successful overtaking under the fully autonomous opposing vehicle with VL-CBF as safety control strategy, where the dimmed and scattered dots are the planned trajectories with $N = 50$, whereas the black solid lines show the actual trajectory the ego vehicle has traveled. The simulation result shows that, when both $veh_f$ and $veh_o$ follow the VL-CBF strategy,  the overall overtaking process is not only safe but also very close to the initially planned trajectories even with perception errors.

Case 2:  $veh_f$ is non-accelerating and $veh_o$ is human driver-based  with an unknown strategy.

In Figure~\ref{fig.sims2_2},   we demonstrate a scenario for this setting, where $veh_o$ is human-driver-based and is assumed to be accelerating towards the speed limit at $\alpha_x^o = 1.6m/s$. In this case, the decision rules in Table~\ref{tab.cons_for_overtaking} is active. As shown in this figure, when $veh_o$ accelerates, its estimated trajectory (blue dots) grew towards the ego vehicle. When ${x_o}_{Q_1}^N \leq {x_e}_{Q_1}^N$, it triggered the  ``Goes Back" condition, and thus, $veh_e$  abandoned the overtaking action.
\begin{figure}[!ht]
    \centering
    \includegraphics[width = 1\linewidth]{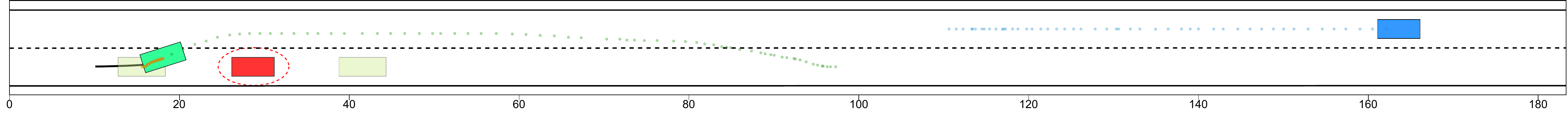}\medskip\\
    \includegraphics[width = 1\linewidth]{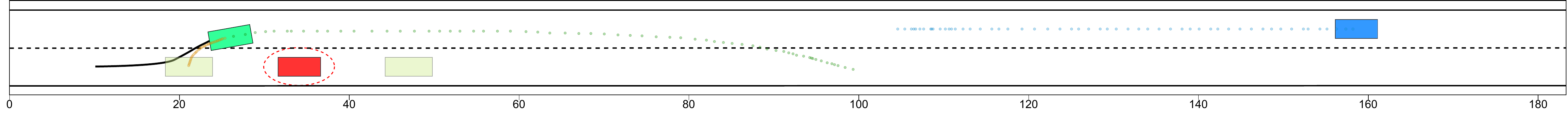}\medskip\\
    \includegraphics[width = 1\linewidth]{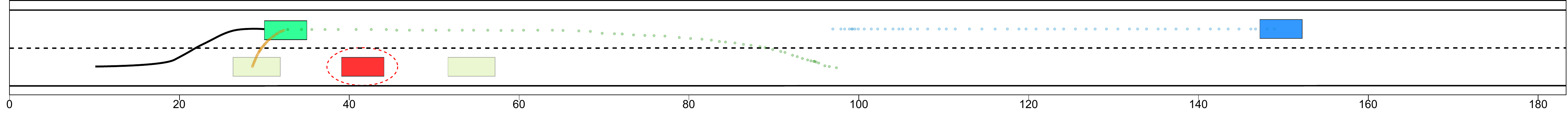}\medskip\\
    \includegraphics[width = 1\linewidth]{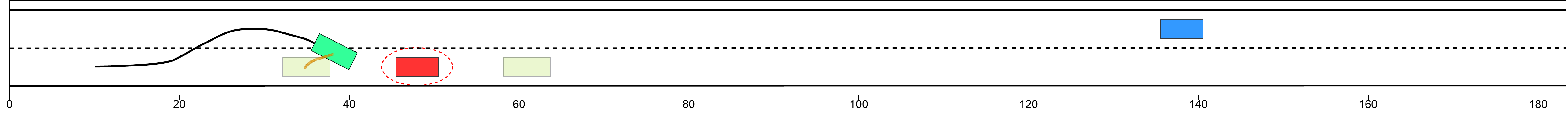}\medskip\\
    \includegraphics[width = 1\linewidth]{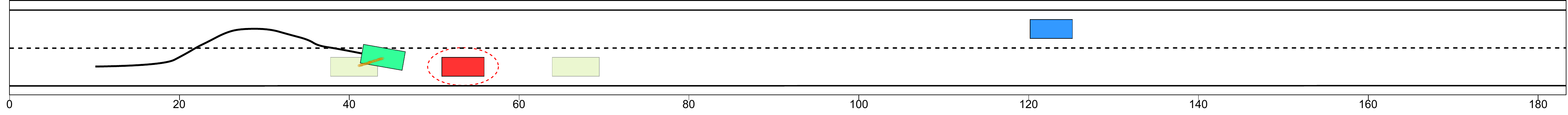}
    \caption{Snapshots from the simulation of the overtaking framework, with $veh_f$ non-accelerating and $veh_o$ human-driver-based and accelerating.}
    \label{fig.sims2_2}
\end{figure}

Case  3: $veh_f$ is accelerating  and $veh_o$ is autonomous with a known VL-CBF strategy.

Note that, although we assumed that the front vehicle is non-accelerating during the overtaking process, our method can still ensure safety by at least returning back to the original position to safely give up the overtaking process, when the front vehicle is accelerating but not interchangeably.  
To see this, we consider three different ranges of accelerating rates for $veh_{f}$. 
For each range, we conduct 30 simulations, where the actual   accelerating rate is picked randomly within the range. 
In  Table~\ref{tab.different_throttles}, for each range, we provide the minimum distance with other vehicles as well as the number of successful overtakings.  Based on the simulation results, we can see that, as the accelerating rate of $veh_f$ increases, the ego vehicle will have less chance to accomplish the overtaking task. 
Nevertheless, the ego vehicle  can still  maintain a safe distance outside the predefined ellipse-shaped safety region even when the overtaking task fails.


\begin{table}[!ht]
    \centering
    \arrayrulecolor{black}
    \caption{\color{black} Overtaking  with Different   Accelerating Rates of $veh_f$.}
    \begin{tabular}{|>{\color{black}}c|>{\color{black}}c|>{\color{black}}c|>{\color{black}}c|}
         \hline
         &  $\alpha_f^x\in(1,3]$ & $\alpha_f^x\in(3,5]$ & $\alpha_f^x\in(5,7]$  \\\hline
         successful overtaking  & 29/30 & 11/30 & 0/30\\\hline
         min dist. with $veh_f$ (m)& 1.27 & 1.24 & 1.31\\\hline
         min dist. with $veh_o$ (m)& 18.31 & 5.52 & 72.41 \\\hline
    \end{tabular}
    \label{tab.different_throttles}
\end{table}

\medskip
In summary, our simulation results in this subsection show the effectiveness of our proposed method under different driving scenarios. 
Specifically, it evaluate the  safe-by-construction property of our framework by assuming that the behaviors of other vehicles are rational \cite{shalev-shwartzFormalModelSafe2017}. 
Therefore, our result  extend  existing safety-critical overtaking maneuvers  \cite{heRuleBasedSafetyCriticalControl2021, heAutonomousRacingMultiple2022} to the general case with the potential opposing traffic on two-way road scenarios. We will show the comparison with the baseline approach in the next section.

\subsection{Comparison with Conventional MPC}

In this subsection, we conduct simulations to further demonstrate the superiority of our approach. 
As discussed in Section~\ref{sec.Related-works-summary},  existing overtaking maneuvers with safety guarantees cannot be directly extended to the case with opposing vehicles.  However, if formal safety guarantee is not required, then the conventional MPC-based method  can be extended to the case with opposing vehicles by further considering the ellipse-shaped distance around the opposing vehicle as new distance constraints, known as MPC distances constraints (MPC-DC) \cite{zengSafetycriticalModelPredictive2021}.  
Therefore, we compare the proposed dual TO-CBF-MPC approach with the standard MPC-DC for some extreme scenarios where the MPC-DC method may not be safe.

Specifically, we consider two simulation  scenarios (A and B), where  the velocities of  $veh_f$ and $veh_o$ are constant during the simulation, i.e., $veh_o$ is assumed to be unknown human driver with constant velocity.  
The velocities of  $veh_f$ are the same in two scenarios, but the velocity of  $veh_f$ in Scenario A is fast than that in Scenario B. 
Therefore,  feasible overtaking actions in Scenario~B may be dangerous in Scenario~A.

For each scenario, as we mentioned above, we use both our TO-CBF-MPC method and the MPC-DC method, 
and their  parameters   are as follows:
\begin{enumerate}[leftmargin=*] 
    \item 
    For the dual TO-CBF-MPC method,
    we still consider the setting in Case~2 as above with $N = 50$,   
    and we investigate two different  VL-CBF levels   $\epsilon = \tilde{\kappa}^{-1}(-\lambda_0) = 0.3$ and $\epsilon =0.5$ respectively;
    \item 
    For the  MPC-DC method, we also set the prediction horizon  as  $N_{\text{MPC}} = 50$, and the distance constraints are set as the ellipse-shaped constraints as in VL-CBF.
\end{enumerate} 

 For each scenario, we also consider the impact of perception errors. Therefore, we repeat the simulation for each scenario by two different methods for 30 times  with perception noises randomly generated by the simulation setups provided in Section~\ref{sec.sim_setups}.

To compare the two approaches, we show the minimum distances with other vehicles as well as the number of safe simulation cases among the 30 cases for each method (one MPC-DC controller and two VL-CBF controllers with different parameters) in Table~\ref{tab.numerical_table}. 
The simulation results show that, 
for  VL-CBF controllers with different levels, our framework can always  guarantee safety under disturbances for all simulations generated. Furthermore, 
for both parameters, the safety margins w.r.t.\   other vehicles are more significant, which means that our method is more robust than the MPC-DC method.

\begin{table}[ht]
    \centering
    \arrayrulecolor{black}
    \caption{\color{black}Minimum distance with surrounding vehicles}
    \begin{tabular}{|>{\color{black}}c|>{\color{black}}c|>{\color{black}}c|>{\color{black}}c|}
        \hline
         (Scenario~A)& MPC-DC & $\epsilon = 0.3$ & $\epsilon = 0.5$ \\\hline
         min dist. with $veh_f$ (m)& -0.15 & 0.76 & 0.99\\\hline
         min dist. with $veh_o$ (m)& 0.31 & 85.66 & 86.07\\\hline
         safe  cases & 0/30 & 30/30 & 30/30\\\hline
         \hline
         (Scenario~B)& MPC-DC & $\epsilon = 0.3$ & $\epsilon = 0.5$ \\\hline
         min dist. with $veh_f$ (m)& 0.37 & 0.77 & 1.02 \\\hline
         min dist. with $veh_o$ (m)& 36.46 & 42.51 & 41.55 \\\hline
          safe  cases & 19/30 & 30/30 & 30/30 \\\hline
    \end{tabular}
    \label{tab.numerical_table}
\end{table}

\begin{figure}
	\centering
	\subfigure[\footnotesize Simulations for a Scenario A case, where our method will safely terminate the overtaking process when it is not feasible.]{\includegraphics[width = \linewidth]{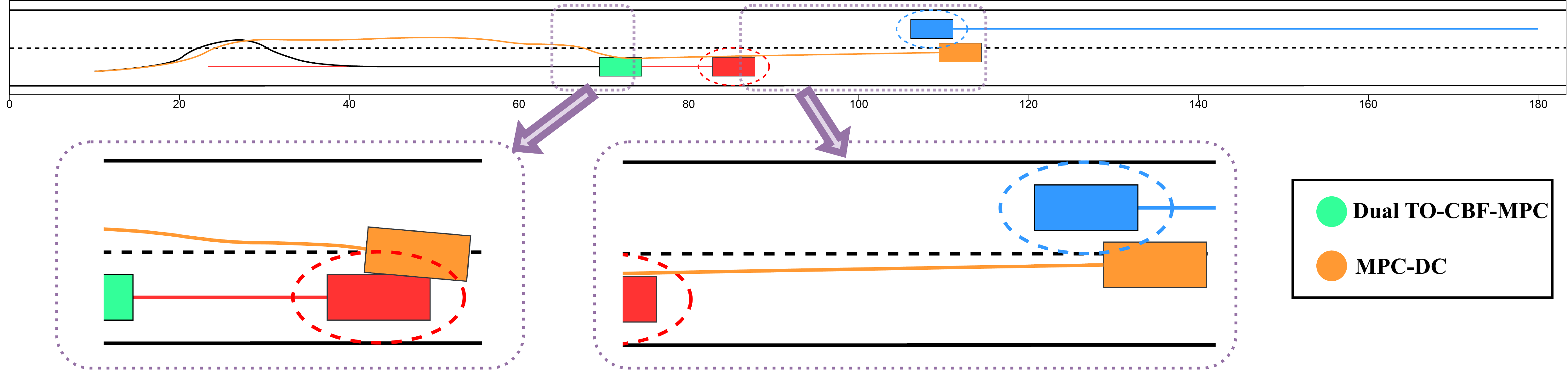} \label{fig.sims4_3}}
	\subfigure[\footnotesize Simulations for a  Scenario~B case, where our method can safely overtake the front vehicle but the existing method fails to do so.]{\includegraphics[width = \linewidth]{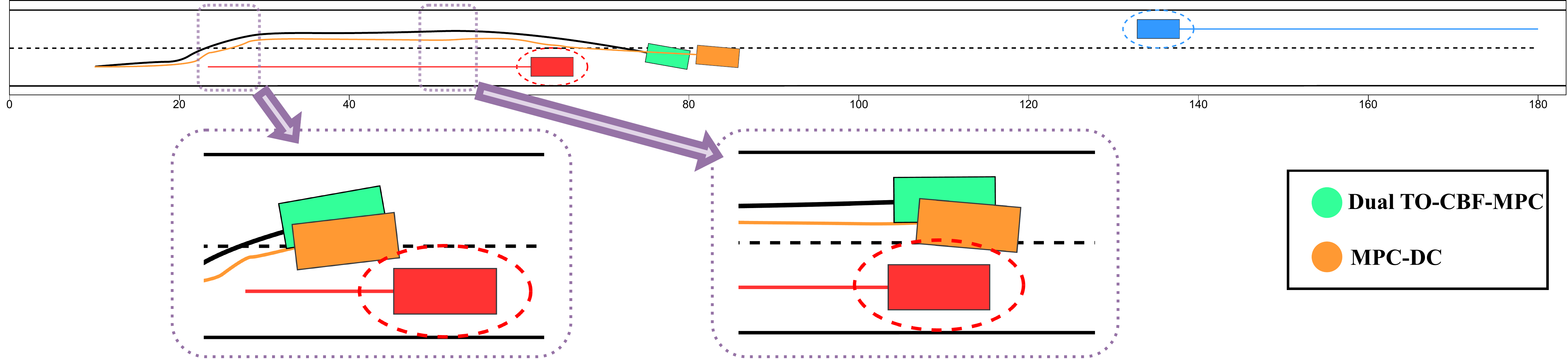} \label{fig.sims4_2}} 
	\vspace{-12pt}
	\caption{Simulations of the proposed dual TO-CBF-MPC framework compared with the  MPC-DC method with opposing traffic.}
	\label{fig.visualized_results}
\end{figure}


In Figure~\ref{fig.visualized_results}, we provide visualized simulation results for a   Scenario~A  case (Figure~\ref{fig.sims4_3}) and  a  Scenario~B  case (Figure~\ref{fig.sims4_3}). 
For Scenario~A in Figure~\ref{fig.sims4_3}, since the opposing vehicle is moving with a fast speed,  the overtaking task is not feasible for both methods. Since the standard MPC-DC does not have a backup strategy, it will collide with the opposing vehicle if it chooses to overtake.  Whereas, our proposed dual TO-CBF-MPC method can automatically choose to terminate the overtaking process, and guide the ego vehicle back to $\mathcal{L}_{\text{ego}}$ safely. For Scenario~B in Figure~\ref{fig.sims4_2}, although there is enough space for the ego vehicle to conduct safe overtaking, and the optimization problem for  MPC-DC  is always feasible, it is still not safe by violating the ellipse-shaped distance constraints at some point due to the inter-sampling effect and the existence of disturbances during the movements. In contrast, the trajectory planned by our proposed dual TO-CBF-MPC method has wider safety margins, and safety is guaranteed even between sampling intervals.


\section{Conclusions}
\label{sec.conclusion}

In this paper, we  introduced a new safe-by-construction control framework to address the  general on-road overtaking tasks for autonomous vehicles with consideration of the potential opposing traffic.  The contributions of our works are summarized as follows. 
First, we proposed  the novel concept of  VL-CBF in which level sets are adjustable, and is particularly suitable for overtaking maneuvers. 
Second, we developed an integrated framework  combining VL-CBF with time-optimal MPC. We show that the  proposed strategy is provably safe, and therefore, it provides a  formal safety guarantee  for our integrated framework. Finally, a set of simulations were conducted to show the effectiveness of our approach for different driving scenarios. Specifically, we show that our approach is applicable for general two-way overall taking problem in the presence of uncertainties, and our approach is more robust compared with the baseline algorithms.

In this work, the safety guarantee is obtained under the assumption that the dynamic model of the vehicle is precise. When the system model is imprecise with bounded errors, our safety guarantee is not theoretically sound. In the future, we would like to further investigate the effect of unmodeled dynamics. Also, we would like to provide probabilistic safety guarantees when the distribution information of the uncertainties is available a priori.

\bibliographystyle{IEEEtran}
\bibliography{main}


\vfill

\end{document}